\documentclass[a4paper,reqno, 12pt]{amsart}

\usepackage[margin=3cm]{geometry}

\makeatletter
\let\old@setaddresses\@setaddresses
\def\@setaddresses{\bigskip\bgroup\parindent 0pt\let\scshape\relax\old@setaddresses\egroup}
\makeatother

\usepackage[utf8]{inputenc}
\usepackage[T1]{fontenc}

\usepackage{amsthm, amssymb, amsmath}

\usepackage{thmtools}
\usepackage{thm-restate}

\newtheorem{theorem}{Theorem}[section]
\newtheorem{lemma}[theorem]{Lemma}
\newtheorem{proposition}[theorem]{Proposition}
\newtheorem{remark}[theorem]{Remark}
\newtheorem{corollary}[theorem]{Corollary}
\newtheorem{observation}[theorem]{Observation}

\usepackage{float}
\usepackage{needspace}
\usepackage{graphicx}
\usepackage{xcolor}
\usepackage{subcaption}

\usepackage{hyperref}
\hypersetup{
  pdftitle = {Local certification of geometric graph classes},
  pdfauthor=  {Oscar Defrain, Louis Esperet, Aurélie Lagoutte, Pat Morin, Jean-Florent Raymond},
  colorlinks = true,
  linkcolor = black!30!red,
  citecolor = black!30!green
}
\usepackage[shortlabels]{enumitem}

\usepackage{marvosym}
\newcommand{\apx}[1]{\hyperref[#1]{\Rightscissors}}

\def\cqedsymbol{\ifmmode$\lrcorner$\else{\unskip\nobreak\hfil
\penalty50\hskip1em\null\nobreak\hfil$\lrcorner$
\parfillskip=0pt\finalhyphendemerits=0\endgraf}\fi}

\newcommand{\formerepsilon}{\kappa}

\title{Local certification of geometric graph classes}

\author[O.~Defrain]{Oscar Defrain}
\address[O.~Defrain]{Aix-Marseille Université, CNRS, LIS, Marseille, France}
\email{oscar.defrain@lis-lab.fr}

\author[L.~Esperet]{Louis Esperet}
\address[L.~Esperet]{Laboratoire G-SCOP (CNRS, Univ.\ Grenoble Alpes), Grenoble, France}
\email{louis.esperet@grenoble-inp.fr}

\author[A.~Lagoutte]{Aurélie Lagoutte}
\address[A.~Lagoutte]{Laboratoire G-SCOP (CNRS, Univ.\ Grenoble Alpes), Grenoble, France}
\email{aurelie.lagoutte@grenoble-inp.fr}

\author[P.~Morin]{Pat Morin}
\address[P.~Morin]{School of Computer Science, Carleton University,
  Canada}
\email{morin@scs.carleton.ca}

\author[J.-F.~Raymond]{Jean-Florent Raymond}
\address[J.-F.~Raymond]{CNRS, LIP, ENS de Lyon, France.}
\email{jean-florent.raymond@cnrs.fr}

\thanks{L.\ Esperet is partially supported by the French ANR Project TWIN-WIDTH
  (ANR-21-CE48-0014-01), and by LabEx
  PERSYVAL-lab (ANR-11-LABX-0025). P.\ Morin is partially supported by
  NSERC. A. Lagoutte and J.-F.~Raymond are partially supported by French ANR project GRALMECO (ANR-21-CE48-0004).}

\begin{document}

\begin{abstract}
  The goal of local certification is to locally convince the vertices of a
  graph $G$ that $G$ satisfies a given property. A prover assigns short
  certificates to the vertices of the graph, then the vertices are
  allowed to check their  certificates and the certificates of their neighbors, and based only
  on this local view and their own unique identifier, they must decide whether $G$ satisfies the given
  property. If the graph indeed satisfies the property, all vertices
  must accept the instance, and otherwise at least one vertex must
  reject the instance (for any possible assignment of
  certificates). The goal is to minimize the size of the certificates.
  
In this paper we study the local certification of  geometric and
topological graph
classes. While it is known that in $n$-vertex graphs, planarity can be certified locally
with certificates of size $O(\log n)$, we show that several closely
related graph classes require certificates of size $\Omega(n)$. This
includes penny graphs, unit-distance graphs, (induced) subgraphs of
the square grid, 1-planar graphs, and unit-square graphs. These bounds
are tight up to a constant factor (or logarithmic factor for unit-square graphs) and give the first known examples
of hereditary (and even monotone) graph classes
for which the certificates must have linear size. For
unit-disk graphs we obtain a lower bound of $\Omega(n^{1-\delta})$
for any $\delta>0$ on the size of the certificates, and an upper bound
of $O(n \log n)$.
The lower bounds are obtained by proving rigidity properties of the
considered graphs, which might be of independent interest.
\end{abstract}

\maketitle

\section{Introduction}

Local certification is an emerging subfield of distributed computing
where the goal is to assign short certificates to each of  the nodes of a
network (some connected graph $G$) such that the nodes can
collectively decide whether $G$ satisfies a given property (i.e., whether it
belongs to some given graph class $\mathcal{C}$) by only inspecting their unique identifier, 
their certificate and the certificates of their neighbors. This assignment of certificates is called a
\emph{proof labeling scheme}, and its \emph{complexity} is the
maximum number of bits of a certificate (as a function of the number of vertices
of $G$, which is usually denoted by $n$ in the paper). If a graph
class $\mathcal{C}$ admits a proof labeling scheme of complexity
$f(n)$, we say that $\mathcal{C}$ has \emph{local complexity}
$f(n)$. Proof labelling schemes are distributed analogues of
traditional non-deterministic algorithms, and graph classes of logarithmic local complexity can be considered as distributed analogues of
classes whose recognition is in \textsf{NP} \cite{Feu21}. The notion
of proof labeling scheme was formally introduced by Korman, Kutten
and Peleg in \cite{KKP10}, but originates in earlier work on
self-stabilizing algorithms (see again \cite{Feu21} for the history of
local certification and a thorough introduction to the field). While
every graph class has local complexity $O(n^2)$
\cite{KKP10},\footnote{Give to every vertex the adjacency matrix of
  the graph and the list of identifiers of all the vertices.} the work of \cite{lcp} identified three natural ranges of
local complexity for graph classes:

\begin{itemize}
\item $\Theta(1)$: this includes $k$-colorability for fixed $k$, and in
  particular bipartiteness;
\item $\Theta(\log n)$: this includes non-bipartiteness and
  acyclicity; and
  \item $\Theta(\text{poly}(n))$: this includes non-3-colorability
    and problems involving
    symmetry.
  \end{itemize}

  It was later proved in \cite{NPY20} that any graph class which can
  be recognized in linear time (by a centralized algorithm) has an ``interactive''
  proof labeling scheme of complexity $O(\log n)$, where
  ``interactive'' means that there are several
  rounds of interaction between the prover
  (the entity which assigns certificates) and the nodes of the network
  (see also \cite{KOS18} for more on distributed interactive protocols). A natural
  question is whether the interactions are necessary or whether such
  graph classes have classical proof labeling schemes of complexity
  $O(\log n)$ as defined above, that is, without multiple rounds of
  interaction.  This question triggered the work of \cite{planar} on planar
  graphs, which have a well-known linear time recognition algorithm. The
  authors of  \cite{planar}  proved that the class of planar graphs indeed has local
  complexity $O(\log n)$, and asked whether the same holds for any
  proper minor-closed class.\footnote{Note that it is easy to show that for any
    minor-closed class $\mathcal{C}$, the complement of $\mathcal{C}$ has local complexity $O(\log n)$, using Robertson and Seymour's Graph Minor Theorem \cite{RS04}.} This
  was later proved for graphs embeddable on any
  fixed surface in \cite{genus} (see also \cite{EL}) and in \cite{BFT} for classes excluding
  small minors, while it was proved in \cite{tw} that classes
  excluding a planar graph $H$ as a minor have local complexity $O(\log^2 n)$. The authors of \cite{tw} also proved the
  related result that any graph class of bounded treewidth which is expressible in second order
  monadic logic has local complexity $O(\log^2 n)$ (this implies in particular that
  for any fixed $k$, the class of graphs of treewidth at most $k$ has
  local complexity $O(\log^2 n)$). Similar meta-theorems involving
  graph classes expressible in some logic were proved
  for graphs of bounded treedepth in \cite{BFT2} and graphs of bounded
  cliquewidth in 
  \cite{fraigniaud2023distributed}.

  Closer to the topic of the present
paper, the authors of \cite{JMR23} obtained proof labeling schemes of
complexity $O(\log n)$ for a number of classes of geometric intersection graphs,
including interval graphs, chordal graphs, circular-arc graphs,
trapezoid graphs, and permutation graphs. It was noted earlier in
\cite{JMR22} (which proved various results on interactive proof
labeling schemes for geometric graph classes) that the ``only'' classes of
graphs for which large lower bounds on the local
complexity are known (for instance
non-3-colorability, some properties involving symmetry~\cite{lcp} or
the diameter~\cite{CPP20}) are
not \emph{hereditary}, meaning that they are not closed under taking
induced subgraphs. 
It turns out that an example of a hereditary class
with polynomial local
complexity had already been identified in \cite{K3F} a couple of years
earlier: triangle-free graphs (the lower bound on the local complexity
given there was sublinear). It was speculated in \cite{JMR22} that
any class of geometric intersection graphs
has
small  local complexity, as such classes are both hereditary and well-structured.

\subsection*{Results}

In this paper we
identify a key rigidity property  in graph classes and use it to
derive a number of \emph{linear} lower bounds on the local complexity of graph classes defined using geometric or topological
properties. These bounds are all best possible, up to $n^{\delta}$
factors, for any $\delta>0$. So 
for a number of classical hereditary
graph classes
studied in structural graph theory, topological graph theory,  and
graph drawing, we show that the local complexity is
$\Theta(n)$. These are the first non-trivial examples of hereditary
classes (some of our examples are even monotone) with linear local
complexity. Interestingly, all the classes we consider are very close
to the class of planar graphs (which is known to have local complexity
$\Theta(\log n)$ \cite{planar,EL}): most of
these classes are either subclasses or superclasses of planar
graphs. Given the earlier results on graphs of bounded
treewidth \cite{tw} and planar graphs, it is natural to try to understand
which sparse graph classes have (poly)logarithmic local complexity. It
would have been tempting to conjecture that any (monotone or
hereditary) graph class of \emph{bounded expansion} (in the sense of
Ne\v set\v ril and Ossona de Mendez \cite{NO})  has polylogarithmic
local complexity, but our results show that this is false, even for very
simple monotone classes of linear expansion.

\medskip

We first show that every class of graphs that contains
at most $2^{f(n)}$ unlabelled graphs of size $n$ has local complexity
$f(n)+O(\log n)$ \footnote{We use the standard assumption that the
  range of unique identifiers of the vertices is $\text{poly}(n)$, see
Section~\ref{sec:defpls} for more details.}. This implies all the upper
bounds we obtain in this paper, as the
classes of graphs we consider usually contain $2^{O(n)}$ or $2^{O(n
  \log n)}$ unlabelled graphs of size $n$.

\medskip

We then turn to lower bounds.
Using rigidity properties in the classes we consider, we give a
$\Omega(n)$ bound on the local
complexity of \emph{penny graphs} (contact
graphs of unit-disks in the plane), \emph{unit-distance graphs}
(graphs that admit an embedding in $\mathbb{R}^2$ where adjacent vertices
are exactly the vertices at Euclidean distance 1), and
(induced) subgraphs of the square grid.
We then consider \emph{1-planar graphs}, which are graphs
admitting a planar drawing in which each edge is crossed by at most
one edge. This superclass of
planar graphs shares many
similarities with them, but we nevertheless prove that it has local
complexity $\Theta(n)$ (while planar graphs have local complexity
$\Theta(\log n)$).

\medskip

Next, we consider \emph{unit-square graphs} (intersection
graphs of
unit-squares in the plane). We obtain a linear lower bound on
the local complexity of triangle-free
unit-square graphs (which are planar) and of unit-square graphs in
general. 
Finally, we consider \emph{unit-disk  graphs} (intersection of unit-disks in
the plane), which are widely used in distributed computing as a model
of wireless communication networks. For this class we reuse some key
ideas introduced in the unit-square case, but as unit-disk graphs are
much less rigid we need to introduce a number of new tools, which
might be of independent interest in the study of rigidity in
geometric graph classes. In particular we answer questions such as: what is
asymptotically the
minimum number of vertices in a unit-disk graph $G$ such that in any
unit-disk embedding of $G$, two given vertices $u$ and $v$ are at
Euclidean distance at least $n$ and at most $n+1$? Or at distance at
least $n$ and at most $n+\varepsilon$, for $\varepsilon\ll n$? Using our
constructions we obtain a lower bound of 
$\Omega(n^{1-\delta})$ (for every $\delta>0$) on the local complexity
of unit-disk graphs. As there are at most $2^{O(n\log n)}$ unlabelled
unit-disk graphs on $n$ vertices~\cite{MCDIARMID2014413}, our first result implies that the
local complexity of unit-disk graphs is $O(n \log n)$, so our results
are nearly tight.

\subsection*{Techniques} 

All our lower bounds are inspired by 
the set-disjointness problem in non-deterministic communication complexity. 
This approach was already used in earlier work in local certification, in
order to provide lower bounds on the local complexity of computing the
diameter \cite{CPP20}, or for certifying non-3-colorability \cite{lcp}. Here the main challenge is to translate the
technique into geometric constraints. 
The key point of the set-disjointness problem is informally the following:
 let $A,B\subseteq \{1,\ldots, N\}$ be the input of some kind of ``two-party system''
that must decide whether $A$ and $B$ are disjoint, given that one party knows
 $A$ and the other knows $B$; then at least $N$ bits of shared (or
exchanged) information are necessary for them to decide correctly. Otherwise,
there are fewer bit combinations than the $2^N$ entries of the form $(A, \overline{A})$, hence the two parties can be fooled to accept a negative
instance built from two particular positive instances sharing the same bit combination.
In the setting of (non-deterministic) communication complexity, the two parties are Alice and Bob; 
in our setting, the two parties will be two subsets of vertices
covering the graph and with small intersection (the intersection must be a small cutset of the whole graph): in the following, we refer to those two
connected subsets of vertices as respectively the ``left'' part and the ``right'' part. 
The ``shared'' bits of information will be the certificates given to their intersection (and to its neighborhood).
To express the sets $A,B$ and their disjointness, the left (resp.~right) part of the graph will be equipped with a path $P_A$ (resp.~$P_B$) of length $\Omega(N)$,
such that $P_A$ and $P_B$ only intersect in their
endpoints.\footnote{We note here that the proof for
  1-planar graphs diverges from this approach, but it is the only
  one.} The crucial rigidity property which we will require is that
in 
any embedding of $G$ as a geometric graph from some class
$\mathcal{C}$, the two paths $P_A$ and $P_B$ will be very close, in
the sense that if $P_A=a_1,\ldots,a_\ell$ and $P_B=b_1,\ldots,b_\ell$,
then $a_i$ is close to $b_i$ for any $1\le i \le \ell$. 
Using this property, we will attach some gadgets 
to the vertices of the path $P_A$ (resp.~$P_B$) depending on $A$ (resp.~$B$),
in such a way that the
resulting graph lies in the class $\mathcal{C}$ if and only if $A$ and $B$
are disjoint. As there is little connectivity 
between the left and the right part,
the endpoints of the paths will have to contain very long
certificates in order to decide whether $A$ and $B$ are disjoint,
hence whether $G\in  \mathcal{C}$ or not.

We present the results in increasing order of difficulty. Subgraphs or
induced subgraphs of infinite graphs such as grids are perfectly rigid
in some sense, with some graphs having unique embeddings up to
symmetry. Unit-square graphs are much less rigid but we can use nice
properties of the $\ell_\infty$-distance and the uniqueness of
embeddings of 3-connected planar graphs. We conclude with unit-disk graphs, which is the least
rigid class we consider. The Euclidean distance misses most of the
properties enjoyed by the $\ell_\infty$-distance and we must work much
harder to obtain the desired rigidity property. 

\subsection*{Outline} We start with some preliminaries on graph classes
and local certification in Section~\ref{sec:prel}. We prove our
general upper bound result in
Section~\ref{sec:linear}. 
Section~\ref{sec:disjcla} introduces the
notion of a \emph{disjointness-expressing} class of graphs,
highlighting the key properties needed 
to derive our local certification lower bounds.
We deduce in Section~\ref{sec:lb} linear
lower bounds on the local complexity of subgraphs of the grid, penny
graphs, and 1-planar graphs. Section~\ref{sec:usg} is devoted for the
linear lower bound on the local complexity of unit-square graphs,
while Section~\ref{sec:udg} contains the proof of our main result, a
quasi-linear lower bound on the local complexity of unit-disk
graphs. We conclude in Section~\ref{sec:ccl} with a number of questions
and open problems.

\section{Preliminaries}\label{sec:prel}

In this paper logarithms are binary, and graphs are assumed to be simple, loopless,
undirected, and connected.
The \emph{length} of a path $P$, denoted by $|P|$, is the number of
edges of $P$.
The \emph{distance} between two vertices $u$ and $v$ in a graph $G$,
denoted by $d_G(u,v)$  is the minimum
length of a path between $u$ and $v$. The \emph{neighborhood} of a vertex $v$
in a graph $G$, denoted by $N_G(v)$ (or $N(v)$ is $G$ if clear from the
context),  is the set of vertices at distance exactly 1 from $v$. The
\emph{closed neighborhood} of $v$, denoted by  $N_G[v]:=\{v\}\cup N_G(v)$,
is the set of vertices at distance at most 1 from $v$. For a set $S$
of vertices of $G$, we define $N_G[S]:=\bigcup_{v\in S}N_G[v]$.

\subsection{Local certification}\label{sec:defpls}

The vertices of any $n$-vertex
graph $G$ are assumed to be assigned distinct (but otherwise arbitrary)
identifiers $(\text{id}(v))_{v\in V(G)}$ from the set 
$\{1,\ldots,\text{poly}(n)\}$. When we refer to a subgraph $H$ of a graph $G$, we
implicitly refer to the corresponding labelled subgraph of $G$.
Note that the identifiers of each of the vertices of $G$ can be stored
using $O(\log n)$ bits.
We follow the
terminology introduced by G\"o\"os and Suomela~\cite{lcp}.

\subsection*{Proofs and provers} A \emph{proof} for a graph $G$ is a function
$P:V(G)\to \{0,1\}^*$ (as $G$ is a  labelled graph,
the proof $P$ is allowed to depend on the identifiers of the
vertices of $G$). The binary words $P(v)$ are called \emph{certificates}. The \emph{size} of $P$ is the maximum size of a
certificate $P(v)$, for $v\in V(G)$.
A \emph{prover} for a graph class $\mathcal{G}$ is a function that
maps every $G\in \mathcal{G}$ to a proof for~$G$.

\subsection*{Local verifiers} A \emph{verifier} $\mathcal{A}$ is a function that takes a
graph $G$, a proof $P$
for $G$, and a vertex $v\in V(G)$ as inputs, and outputs an element of
$\{0,1\}$. We say that $v$ \emph{accepts} the instance if
$\mathcal{A}(G,P,v)=1$ and that $v$ \emph{rejects} the instance if
$\mathcal{A}(G,P,v)=0$.

Consider a graph $G$, a proof $P$ for $G$, and a
vertex $v\in V(G)$.
We denote by $G[v]$ the
subgraph of
$G$ induced by $N[v]$, the closed neighborhood of $v$, and similarly we denote by $P[v]$ the restriction of $P$ to
$N[v]$.

A verifier $\mathcal{A}$ is \emph{local} if for any $v\in G$, the
output of $v$ only depends on its identifier and $P[v]$.

\smallskip

\noindent \emph{Remark.} Note that our lower bounds hold in the stronger model of \emph{locally
checkable proofs} of G\"o\"os and Suomela~\cite{lcp}, where in
addition the output
of $v$ is allowed to depend on $G[v]$, that is 
$\mathcal{A}(G,P,v)=\mathcal{A}(G[v],P[v],v)$
for any vertex $v$ of $G$.

\subsection*{Proof labeling schemes}

A \emph{proof labeling scheme}
for a graph class $\mathcal{G}$ is a prover-verifier pair
$(\mathcal{P},\mathcal{A})$ where $\mathcal{A}$ is local, with the following properties.

\medskip

\noindent {\bf Completeness:} If $G\in \mathcal{G}$, then
$P:=\mathcal{P}(G)$ is a proof for $G$ such that for any vertex $v\in
V(G)$, $\mathcal{A}(G,P,v)=1$.

\medskip

\noindent {\bf Soundness:}  If $G\not\in \mathcal{G}$, then for every proof
$P'$ for $G$, there exists a vertex $v\in
V(G)$ such
that  $\mathcal{A}(G,P',v)=0$.

\medskip

In other words, upon looking at its closed neighborhood (labelled by
the identifiers and certificates), the local verifier of each vertex
of a graph $G\in \mathcal{G}$ accepts the instance, while if $G\not\in
\mathcal{G}$, for every possible choice of certificates, the local verifier of at least one vertex rejects the instance.

\medskip

The \emph{complexity} of the proof labeling scheme is the maximum size of a
proof $P=\mathcal{P}(G)$ for an $n$-vertex graph $G\in\mathcal{G}$,
and the \emph{local complexity} of $\mathcal{G}$ is the minimum
complexity of a proof labeling scheme for $\mathcal{G}$. If we say that the complexity is $O(f(n))$, for some
function $f$, the $O(\cdot)$ notation refers to $n\to \infty$. See
\cite{Feu21,lcp} for more details on proof labeling schemes and local
certification in general.

\subsection{Geometric graph classes}\label{sec:gc}

In this section we collect some useful properties that are shared by
most of the graph classes we will investigate in the paper.

\medskip

A \emph{unit-disk graph} (respectively \emph{unit-square graph}) is the intersection graph of unit-disks (respectively unit-squares) in the plane. That is, $G$ is a unit-disk graph if every vertex of $G$ can be mapped to a unit-disk in the plane so that two vertices are adjacent if and only if the corresponding disks intersect, and similarly for squares.
A \emph{penny graph} is the contact graph of unit-disks in the plane,
i.e., in the definition of unit-disk graphs above we additionally
require the disks to be pairwise interior-disjoint. A
\emph{unit-distance} graph is a graph whose vertices are points in the
plane, where two points are adjacent if and only if their  Euclidean
distance is equal to 1. Unit-distance graphs clearly form a superclass of penny graphs.

\medskip

A \emph{drawing} of a graph $G$ in the plane is a mapping from the
vertices of $G$ to distinct points in the plane and from the edges of
$G $ to
Jordan curves, such that for each edge $uv$ in $G$, the curve
associated to $uv$ joins the images of $u$ and $v$ and does not
contain the image of any other vertex of $G$. A graph is \emph{planar}
if it has a drawing in the plane with no edge crossings (such a
drawing will also be called a \emph{planar graph drawing} in the
remainder). Every planar graph drawing of a graph $G$ gives a clockwise cyclic
ordering of the neighbors around each vertex of $G$. We say
that two planar graph drawings of $G$ are \emph{equivalent}
if the corresponding cyclic orderings are the same. A \emph{planar
  graph embedding} of a graph $G$ is an equivalence class of planar
graph drawings of $G$. Given a planar graph embedding of a graph $G$, all the
corresponding (equivalent) planar drawings of $G$ have the same set of
faces (but different choices of outerface yield different
planar drawings).   

\medskip

A graph is \emph{1-planar} if it has a drawing in the plane 
such that
for each edge $e$ of $G$, there is at most one edge $e'$ of $G$ distinct from $e$ such that the interior of
the curve associated to $e$ intersects the interior of
the curve associated to $e'$.

\medskip

\begin{figure}[htb]
  \centering
  \includegraphics[scale=0.85]{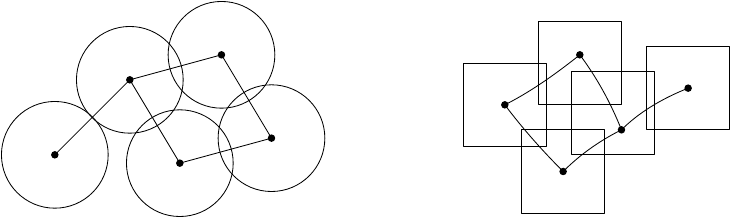}
  \caption{Triangle-free intersection graphs of unit-disks and
    unit-squares in the plane,
    and the associated planar graph embeddings.}
  \label{fig:inter}
\end{figure}

The following well-known proposition will be useful (see Figure~\ref{fig:inter} for an illustration). 

\begin{proposition}\label{pro:trifree}
Consider a family of unit-disks or a family of unit-squares in the
plane, and assume that the intersection graph $G$ of the family is
connected and 
triangle-free. Then $G$  is planar, and moreover each representation of $G$ as such
an intersection graph of unit-disks or unit-squares in the plane  gives rise to a planar graph
embedding of $G$ in
a natural way (see for instance Figure~\ref{fig:inter}). Furthermore, the
 representation of $G$ as an intersection graph (of unit-disks or unit-squares) and the resulting planar graph embedding are equivalent, in the sense that the clockwise cyclic ordering of the neighbors around
 each vertex is the same.
\end{proposition}

\begin{proof}
  Let $R_1,\ldots,R_n$ be the family of unit-disks or the family of unit-squares.
Since $G$ is triangle-free and connected we can assume that the
$R_i$'s are pairwise distinct (since otherwise $G$ consists of a
single edge, and the statement certainly holds). Since $G$ is not a
single edge and is triangle-free, we can choose
one point $v_i$ in each region $R_i$ that is not included in any
region $R_j$ for $j\ne i$, and for any two
intersecting regions $R_i$ and $R_j$, choose a point $r_{ij}=r_{ji}$ in $R_i \cap
R_j$. Inside each region $R_i$, join $r_i$ to all the points
$r_{ij}$ by internally disjoint Jordan arcs. We obtain a planar
drawing of $G$, and all planar drawings that we can obtain in this way
(starting with  the family $R_1, \ldots, R_n$) are equivalent: the clockwise cyclic ordering of the neighbors around
each vertex is the same  in all such drawings (it corresponds to
the clockwise ordering of the intersections with other regions along
the boundary of a given region).
\end{proof}

We will often need to argue that some planar graphs have unique
planar embeddings.
We say that a graph is \emph{3-connected} if it is connected and remains so after the deletion of any two vertices.
The following classical result of Whitney will be crucial.

\begin{theorem}[\cite{Whi32}]\label{thm:whi}
If a planar graph $G$ is 3-connected (or can be
obtained from a 3-connected simple graph by subdividing some edges),
then it has a unique planar graph embedding, up to the reversal of all
cyclic orderings of neighbors around the vertices.
\end{theorem}

We note that the reversal of all cyclic orderings in the statement of
Theorem~\ref{thm:whi} corresponds to a
reflection of the corresponding planar drawings. 

\section{Linear upper bounds for tiny classes}\label{sec:linear}

Given a class of graphs $\mathcal{C}$ and a positive integer $n$, let
$\mathcal{C}_n$ be the set of all unlabelled graphs of $\mathcal{C}$
having exactly $n$ vertices (i.e., we consider graphs up to
isomorphism).

\medskip

If there is a constant $c>0$ such that for every positive integer $n$,
$|\mathcal{C}_n|\le c^n$, then the class $\mathcal{C}$ is said to be
\emph{tiny}. This is the case for all
proper minor-closed classes (for instance planar graphs), and more
generally any class of bounded twin-width (for instance 1-planar graphs) according to \cite{bonnet2021twin}. It is also easy to show
that  for any finitely generated group $\Gamma$ and any finite
set of generators $S$, the class of finite  subgraphs of
$\mathrm{Cay}(\Gamma, S)$ is tiny (this is proved in \cite{bonnet2021twin} for induced subgraphs, but the result for subgraphs follows
immediately as these graphs have $O(n)$ edges).
On the other hand, unit-interval graphs and unit-disk graphs do not form tiny classes as proved in~\cite{MCDIARMID2014413}.

\begin{theorem}
  Any class $\mathcal{C}$ of connected graphs has local complexity at most
$\log(|\mathcal{C}_n|)+O(\log n)$. In particular if  $\mathcal{C}$ is
a tiny class, then the local complexity is $O(n)$.
\end{theorem}

\begin{proof}
Let $G\in \mathcal{C}_n$. The certificate given by the prover to
each vertex $v$ of $G$ contain the following:

\begin{itemize}
\item the number of vertices $n$;
\item a $\log(|\mathcal{C}_n|)$-bit word $w$ representing a graph $G'$ isomorphic to $G$;
\item the name of the vertex $\pi(v)$ corresponding to $v$ in $G'$.
\end{itemize}
In addition to this, the vertices of $G$ store a locally certified spanning
tree $T$, rooted in some vertex $r\in G$, which can be done with
$O(\log n)$ additional bits per vertex (this also encodes the
parent-child relation in the tree $T$, so that each vertex knows that
it is the root $r$ or knows the identifier of its parent in the
tree). We also give to each vertex $v$ the number of vertices in its
rooted subtree $T_v$.
In total, the certificates
above take $\log(|\mathcal{C}_n|)+O(\log n)$ bits, as desired.

\medskip

We now describe the verifier part. Each
vertex checks that it has been given the same value of $n$ and the same
word $w$ describing some graph $G'\in \mathcal{C}$ as its neighbors. The spanning tree $T$ is
then used to compute the number of vertices of $G$ and the root $r$
checks that this number coincides with $n$ and the number of vertices
of $G'$ (this is standard: each
vertex $v$ checks that the number of vertices in its rooted subtree
$T_v$, which was given as a certificate, is equal to the sum of the
number of vertices in the rooted subtrees of its children in $T$, plus
1).
Then each vertex $v$
verifies that $\pi$ is a local isomorphism from $G$ to $G'$, that is,
$\pi$ maps bijectively the neighborhood of $v$ in $G$ to the neighborhood of
$\pi(v)$ in $G'$.

\medskip

We now analyze the scheme. If $G\in \mathcal{C}$, then clearly all the
vertices accept. Assume now that all the vertices accept. Then $\pi$
is a local isomorphism from $G$ to some graph $G'\in \mathcal{C}$,
with the same number of vertices as $G$. As $G'$ is connected, $\pi$
is surjective, but as $G$
and $G'$ have the same number of vertices, $\pi$ must also be
injective. Thus $\pi$ is a bijection and $G$ and $G'$ are isomorphic,
which implies that $G\in \mathcal{C}$.
\end{proof}

As a consequence, we immediately obtain the following.

\begin{corollary}\label{cor:linear}
  The following classes have local complexity $O(n)$:
  \begin{itemize}
    \item the class of all (induced) subgraphs of the square grid,
      \item any class of bounded twin-width,
    \item penny graphs,
    \item 1-planar graphs,
    \item triangle-free unit-square graphs, and
      \item triangle-free unit-disk graphs.
  \end{itemize}
\end{corollary}

\begin{proof}
The fact that the classes in the first two items are  tiny is proved
in \cite{bonnet2021twin}. All the other classes in the statement have bounded twin-width
(the classes in the final two items are planar, by Proposition~\ref{pro:trifree}).
\end{proof}

The next result directly follows from a bound of order $2^{O(n\log n)}$ on the
number of unit-square graphs and unit-disk graphs
\cite{MCDIARMID2014413}, and on the number of unit-distance graphs \cite{alon2014two}.

\begin{corollary}\label{cor:nlogn}
  The classes of unit-distance graphs, unit-square graphs, and unit-disk graphs have local complexity $O(n\log n)$.
\end{corollary}

The remainder of the paper consists in proving lower bounds of order
$\Omega(n)$ (or $\Omega(n^{1-\delta})$, for any $\delta>0$), for all
the classes mentioned in Corollaries~\ref{cor:linear} and
\ref{cor:nlogn}, except the second and last items of Corollary~\ref{cor:linear}.

\section{Disjointness-expressing graph classes}
\label{sec:disjcla}

In this section we describe the framework relating the
disjointness problem to
proof labeling schemes. 
Our main source of inspiration is \cite{lcp},  where a lower bound
on the local complexity of non-3-colorability is proved using a
similar approach, and \cite{CPP20} where an explicit reduction to the
non-deterministic communication complexity of the disjointness
problem is used. 

A class $\mathcal{C}$ of graphs is said to be
$(s, \formerepsilon)$-\emph{disjointness-expressing}
if for some constant ${\alpha>0}$, for every positive integer
$N$ and every $X \subseteq \{1,\ldots, N\}$, one can define graphs
$L(X)$ (referred to as the ``left part'') and $R(X)$ (``right part''), each containing a labelled
set $S$ of special vertices 
such that for every $A,B\subseteq \{1,\ldots, N\}$ the following holds:

\begin{enumerate}[label=(\roman*)]
\item\label{item:total-size} the graph $g(L(A),R(B))$ obtained by identifying vertices of $S$
  in $L(A)$ to the corresponding vertices of $S$ in $R(B)$ is
  connected and has at most $\alpha N^{1/\formerepsilon}$ vertices;
\item\label{item:indep-size} the subgraph of $g(L(A), R(B))$ induced by 
  $N_{g(L(A), R(B))}[S]$ %
  is independent\footnote{for all
    $A,A',B,B'\subseteq \{1, \ldots, N\}$ there is an isomorphism from
    $g(L(A),R(B))[N_{g(L(A), R(B))}[S]]$ to
    $g(L(A'),R(B'))[N_{g(L(A'), R(B'))}[S]]$ that is the identity on $S$.} of
  the choice of $A$ and $B$ and has at most $s$ vertices; and
\item\label{item:no-intersection} $g(L(A),R(B))$ belongs to $\mathcal{C}$ if and only if $A\cap
  B=\emptyset$.
\end{enumerate}

The idea is that $S$ is a small cutset between vertices of $L(A)$, having information on $A$, and vertices of $R(B)$, having information on $B$. Deciding whether the graph belongs to $\mathcal{C}$ amounts to deciding whether $A$ and $B$ are disjoint, which requires $N$ bits of information even in a non-deterministic setting, thus the small cutset at the frontier between $L(A)$ and $R(B)$ must receive long certificates. Otherwise,
there are fewer bit combinations at the frontier than the $2^N$ entries of the form $(A, \overline{A})$, hence the vertices can be fooled to accept a negative
instance built from two particular positive instances sharing the same bit combination.

We note here that in a previous version of this manuscript, the proof
of Theorem~\ref{thm:dex} was 
using an explicit reduction from the non-deterministic communication complexity
of the set-disjointness problem. In this version we have opted for a
more direct argument, inspired by the proof of \cite[Theorem
6.4]{lcp}.

\medskip

The role of $s$ and $\formerepsilon$ from the definition of $(s, \formerepsilon)$-disjointness-expressing
is explained by
the result below.

\begin{theorem}\label{thm:dex}
Let $\mathcal{C}$ be a $(s,\formerepsilon)$-disjointness-expressing class
of graphs. Then any proof labeling scheme for the class $\mathcal{C}$
has complexity $\Omega \left (\frac{n^{\formerepsilon}}{s} \right )$. In particular if
$s$ is a constant and $\formerepsilon=1$, the complexity is $\Omega(n)$.
\end{theorem}

\begin{proof}
  Let $(\mathcal{P}, \mathcal{A})$ be a proof labeling scheme for the
  class $\mathcal{C}$ and let $p\colon \mathbb{N} \to \mathbb{N}$ be a
  monotone upper-bound on its complexity.
  Let $N$ be a positive integer. For every $A\subseteq \{1, \ldots,
  N\}$, let $\overline{A} = \{1, \dots, N\} \setminus A$ and let $G_A =
  g(L(A), R(\overline{A}))$.
  Clearly $G_A \in \mathcal{C}$ so the verifier $\mathcal{A}$ accepts
  the proof $P_A = P(G_A)$ on every vertex of $G_A$. Let $n$ denote
  the maximum number of vertices of $G_A$ for $A\subseteq \{1, \ldots,
  N\}$.
  
  There are $2^N$ choices for the set $A$. On the other hand,
  in a graph on at most $n$ vertices, the certificates given by $\mathcal{P}$
  have at most $p(n)$ bits so on a subset of $s$ vertices there are at most
  $2^{s\cdot p(n)}$ choices of certificates in total. In particular in
  $G_A$ there are at most $2^{s\cdot p(n)}$ different ways to assign
  certificates to the vertices of $N[S]$.

  By the Pigeonhole Principle, if $2^N > 2^{sp(n)}$ there are two sets
  $A, A'\subseteq \{1, \dots, N\}$ such that the proofs $P_A$ and $P_{A'}$
  coincide on the subgraph of $G_A$ and $G_{A'}$ induced by
  $N[S]$. Recall that by definition this subgraph does not depend on
  the choice of $A$.
  Since $A\ne A'$, we have $A \cap \overline{A'} \neq \emptyset$ or $A'
  \cap \overline{A} \neq \emptyset$. By symmetry we may assume without loss
  of generality that we are in the first case. So the graph $G =
  g(L(A), R(\overline{A'}))$ does not belong to $\mathcal{C}$.
  
  We now consider
  a proof $P$ for $G$ defined as follows: if $v \in V(L(A))$ then
  $P(v) := P_A(v)$ and if $v \in V(R(\overline{A'}))$ then $P(v) :=
  P_{A'}(v)$ (as $P_A$ and $P_{A'}$ coincide on $L(A)\cap
  R(\overline{A'})=S\subseteq N[S]$, the proof $P$ is well defined).
  Observe that if $v \in V(L(A))$, then $(G[v], P[v]) = (G_A[v],
  P_A[v])$ and  if $v \in V(R(\overline{A'}))$, then $(G[v], P[v]) = (G_{A'}[v], P_{A'}[v])$. So
  the verifier $\mathcal{A}$ accepts $P$ on every vertex of $G$, which
  contradicts the fact that $G$ does not belong to $\mathcal{C}$. 

  Therefore $2^N\leq 2^{s\cdot p(n)}$. Recall that
  $n\leq \alpha N^{1/\formerepsilon}$, by the definition of
  disjointness-expressibility.
  Hence $p(n) = \Omega(n^\formerepsilon/s)$, as claimed.  
\end{proof}

In general a lower bound on the complexity of a proof
labeling scheme for a graph class does not immediately translate to
results for sub- or super-classes. This can be compared to what happens in
centralized algorithms, where the computational hardness of the
recognition problem
for a graph class does not imply in general that a similar result holds
for sub- or super-classes.
Sometimes however the proof that a graph class is
disjointness-expressing also provides results for sub- or
super-classes, as described below.

\begin{remark}
\label{rmk:disjointness implied by other class}
Let $\mathcal{C}$ be a class of graphs and let $\mathcal{C}^- \subseteq \mathcal{C}$ be a subclass of $\mathcal{C}$.
\begin{enumerate}
\item Assume $\mathcal{C}^-$
 is $(s, \formerepsilon)$-disjointness-expressing, as witnessed by functions $L$ and $R$ as in the definition. Suppose furthermore that for every $A,B\subseteq \{1, \ldots, N\}$ such that $A\cap B \neq \emptyset$, $g(L(A),R(B)) \notin \mathcal{C}$. Then $\mathcal{C}$ is $(s, \formerepsilon)$-disjointness-expressing.\label{item: get bigger class disjointness-expressing}
 \item Assume $\mathcal{C}$
is $(s, \formerepsilon)$-disjointness-expressing, as witnessed by functions $L$ and $R$ as in the definition. Suppose furthermore that for every $A,B\subseteq \{1, \ldots, N\}$ such that $A\cap B = \emptyset$, $g(L(A),R(B)) \in \mathcal{C}^-$. Then $\mathcal{C}^-$ is $(s, \formerepsilon)$-disjointness-expressing.\label{item: get smaller class disjointness-expressing}
 \end{enumerate}
\end{remark}

\section{Linear lower bounds in rigid classes}\label{sec:lb}

In this section we obtain linear lower bounds on the local complexity of several graph classes using the framework described in Section~\ref{sec:disjcla}.

\subsection{Penny graphs and unit-distance graphs}

\begin{theorem}\label{thm:penny}
The class of penny graphs is $(6,1)$-disjointness-expressing.
\end{theorem}

\begin{proof}
  The construction is described in Figure~\ref{fig:penny}. Let
  us argue that for $A,B\subseteq \{1, \ldots, N\}$ there is a unique
  (up to reflection, translation, and rotation) penny representation
  of $L(A), R(B)$, and $g(L(A), R(B))$ (shown on Figure~\ref{fig:penny
    glued}).  
  We can observe that for each of $L(A), R(B)$, and $g(L(A), R(B))$, there exists an ordering $\{v_1,v_2,
  \ldots\}$ of the vertices such that $v_1, v_2, v_3$ is
  a triangle, and each $v_i$ (for $i>3$) has two neighbors $x$ and $y$
  in $\{v_1, \ldots, v_{i-1}\}$ such that $xyv_{i_0}$ is a triangle
  for some $i_0<i$. Once we fix the image of the first triangle
  $v_1v_2v_3$ in the plane (which must
  form a unit equilateral triangle), there is a unique way to embed
   $v_i$ in the plane: its image must be at distance
  exactly one of the images of $x$ and $y$. This condition  is satisfied by
  exactly two points in the plane, one of which is already used by
  $v_{i_0}$.

The set $S=\{c_1, c_2\}$ has size 2 and we can observe that the
subgraph induced by the neighborhood of $S$ in $g(L(A), R(B))$ is
independent from the choice of $A$ and $B$, and has size at most 6. Hence Condition~\ref{item:indep-size} of disjointness-expressing is satisfied.
Moreover $g(L(A), R(B))$ has at most $18N+12$ vertices ensuring Condition~\ref{item:total-size} with $\formerepsilon=1$.
Finally, we can see on Figure~\ref{fig:penny glued} that $a_i$ and
$b_i$ cannot both exist at the same time since otherwise their images
in the plane would coincide.
So if there exists $i\in A\cap B$, $a_i$ and $b_i$ must both exist and
$g(L(A), R(B))$ is not a penny graph. On the other hand, if $A\cap B =\emptyset$ then at most one of $a_i, b_i$ exists for every $i$ and $g(L(A), R(B))$ is a penny graph with representation given in Figure~\ref{fig:penny}. So Condition~\ref{item:no-intersection} of disjointness-expressing is satisfied.
\end{proof}

\begin{figure}[h]

\begin{subfigure}[t]{0.23\linewidth}
\centering
\includegraphics[scale=1, page=1]{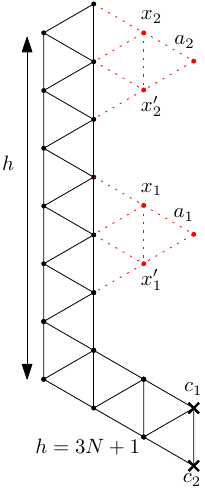}
\subcaption{$L(A)$, where the triangle $a_ix_ix'_i$ exists if and only
  if $i\in A$}
\end{subfigure}\hfill
\begin{subfigure}[t]{0.23\linewidth}
\centering
\includegraphics[scale=1, page=2]{Penny}
\subcaption{$R(B)$, where the triangle $b_iy_iy'_i$ exists if and only
  if $i\in B$}
\end{subfigure}\hfill
\begin{subfigure}[t]{0.45\linewidth}
\centering
\includegraphics[scale=1, page=3]{Penny}
\subcaption{$g(L(A), R(B))$ obtained by identifying $c_1$ and $c_2$. Grey circles show the only possible position for a penny representation of the graph, with a collision between $a_i$ and $b_i$. }
\label{fig:penny glued}
\end{subfigure}

\caption{Construction of $L,R$ and $g$ for penny graphs in the case where $N=2$, with $A,B\subseteq\{1, \ldots, N\}$. Color red highlights vertices and edges that depend on the choice of $A$, and color blue highlights vertices and edges that depend on the choice of $B$.}

\label{fig:penny}
\end{figure}

From Theorems~\ref{thm:penny} and \ref{thm:dex}, together with Corollary~\ref{cor:linear}, we immediately deduce the following.

\begin{theorem}
The local complexity of the class of penny graphs is $\Theta(n)$.
\end{theorem}

The exact same construction as in the proof of Theorem~\ref{thm:penny}
can be used to deal with unit-distance graphs.
Indeed, observe that the same rigidity arguments show that
$g(L(A), R(B))$ is a unit-distance graph if and only if $A$ and $B$
are disjoint. Hence by the first item of Remark~\ref{rmk:disjointness
  implied by other class} we get the following.

\begin{theorem}
The class of unit-distance graphs is $(6,1)$-disjointness-expressing.
\end{theorem}
As above we have the following consequence.
\begin{theorem}
The local complexity of the class of unit-distance graphs is $\Omega(n)$.
\end{theorem}

We observe that rigidity properties similar to those used in the proof
of Theorem~\ref{thm:penny} can be obtained in higher dimension $d\geq
3$ with a very similar construction. This suggests that via a very
similar proof one can obtain a linear lower
bound on the local complexity of contact graphs of balls and
unit-distance graphs in dimension $d$, for any~$d\geq 3$.

\subsection{Subgraphs of the square grid}

\begin{figure}[h]

\begin{subfigure}[t]{0.25\linewidth}
\centering
\includegraphics[scale=0.7, page=1]{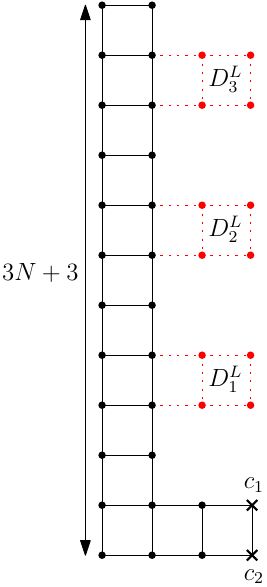}
\subcaption{$L(A)$, where the block $D^L_i$ (4 vertices, 6 edges) exists if and only if $i\in A$}
\end{subfigure}
\hspace{15pt}
\begin{subfigure}[t]{0.25\linewidth}
\centering
\includegraphics[scale=0.7, page=2]{squareGrid}
\subcaption{$R(B)$, where the block $D^R_i$ (4 vertices, 6 edges) exists if and only if $i\in B$}
\end{subfigure}
\hspace{15pt}
\begin{subfigure}[t]{0.3\linewidth}
\centering
\includegraphics[scale=0.7, page=3]{squareGrid}
\subcaption{$g(L(A), R(B))$ obtained by identifying $c_1$ and $c_2$.}

\end{subfigure}

\captionsetup{width=.85\linewidth} 
\caption{Construction of $L,R$ and $g$ for subgraphs of the square grid in the case where $N=3$, with $A,B\subseteq\{1, \ldots, N\}$. Color red highlights vertices and edges that depend on the choice of $A$, and color blue highlights vertices and edges that depend on the choice of~$B$.}
\label{fig:squareGrid}
\end{figure}

\begin{theorem}\label{thm:squareGrid}
The class of subgraphs of the square grid is $(6, 1)$-disjointness-expressing.
\end{theorem}

\begin{proof}
The construction is described  in Figure~\ref{fig:squareGrid}, and is very similar to the one used for penny
graphs in Theorem~\ref{thm:penny}. We denote by $D^L_i$ a
\emph{left-truncated domino} containing 4 vertices and 6 edges,
attaching to 2 existing vertices on the left, as shown by one red
block on the figure. We add such a subgraph $D^L_i$  in $L(A)$ if and
only if $i$  belongs to $A$. Similarly we define a
\emph{right-truncated domino} $D^R_i$  containing 4 vertices and 6 edges, attaching to 2 existing vertices on the right, as shown by one blue block on the figure, and $D^R_i$ is present in $R(B)$ if and only if $i$ belongs to $B$.
The set $S=\{c_1, c_2\}$ has size 2 and we can observe that the
subgraph induced by the neighborhood of $S$ in $g(L(A), R(B))$ is
independent from the choice of $A$ and $B$ and has size at most 6, which fulfills Condition~\ref{item:indep-size} of disjointness-expressing.
Moreover the size of $g(L(A), R(B))$ is at most $20N+18$, satisfying Condition~\ref{item:total-size} with $\formerepsilon=1$.
Finally, we claim that if $g(L(A), R(B))$ is a subgraph of the square
grid, the blocks $D^L_i$ and $D^R_i$ cannot both exist because there
is not enough space in the grid to fit two different vertices at their
extremities, in a sense that we explain now. Observe that $g(L(A),
R(B))$ can be constructed by gluing $C_4$'s along their edges. As
every edge of the square grid is shared by exactly two $C_4$'s, as
soon as we embed one $C_4$ of $g(L(A), R(B))$ as an induced
subgraph of the square grid, there is at most one way to extend this to
an embedding of $g(L(A), R(B))$  as an induced subgraph of the square grid.
If both $D^L_i$ and $D^R_i$ exist in $g(L(A), R(B))$ and this graph is
an induced subgraph of the grid, the aforementioned rigidity property implies that two vertices of the grid belong to both of the truncated dominos, which is impossible since these subgraphs are disjoint in $g(L(A), R(B))$. So in this case, $g(L(A), R(B))$ is not a subgraph of the grid. Conversely it is easy to check that when $A$ and $B$ are disjoint, $g(L(A), R(B))$ is indeed an induced subgraph of the grid. This shows Condition~\ref{item:no-intersection}.
\end{proof}

From the proof above, we deduce the following result for induced
subgraphs of the square grid.

\begin{corollary}\label{coro:induced subgraphs of squareGrid}
The class of induced subgraphs of the square grid is $(6, 1)$-disjointness-expressing.
\end{corollary}

\begin{proof}
We observe that the graph $g(L(A), R(B))$ is an induced subgraph of the square grid whenever $A\cap B=\emptyset$. Therefore by Remark~\ref{rmk:disjointness implied by other class}, Item~\ref{item: get smaller class disjointness-expressing}, the class of induced subgraphs of the square grid is disjointness-expressing with the same parameters as subgraphs of the square grids.
\end{proof}

From Theorem~\ref{thm:squareGrid}, Corollary~\ref{coro:induced subgraphs
  of squareGrid}, and Theorem~\ref{thm:dex}, together with Corollary~\ref{cor:linear}, we immediately deduce the following.

\begin{theorem}
The local complexity of the class of (induced) subgraphs of the square
grid is $\Theta(n)$.
\end{theorem}

Using similar techniques, we can prove that the same holds for grids in any fixed dimension
$d\ge 2$.

\subsection{1-planar graphs}

\begin{theorem}\label{thm:1planar}
The class of 1-planar graphs is $(20, 1)$-disjointness-expressing.
\end{theorem}

\begin{proof}
Figure~\ref{one_planar_lowerbound_fig} illustrates the graph used in
the proof. The fact that the graph of
Figure~\ref{one_planar_lowerbound_fig}\textsc{(c)} (even without any
of the edges $a_ia_i'$ and $b_jb_j'$) has a unique
1-planar embedding follows from a result of \cite{korzhik:minimal} (about the outer
ring of vertices, which is $C_{2N + 6}\boxtimes P_4$). Note that all
edges except those on the 2-edge paths between consecutive vertices
$b_j,b_{j+1}$ are crossed exactly once: a pigeonhole argument shows
that any other drawing inside the region bounded by the outer ring
would have an edge crossed twice. This forces each such
2-edge path between $b_j$ and $b_{j+1}$ to lie inside the region bounded by
$a_j,b_j,a_j',b_{j+1}$. We emphasize that this is not true if we
restrict ourselves to the
graph of Figure~\ref{one_planar_lowerbound_fig}\textsc{(b)}, where the
2-edge paths might leave the region bounded by $b_j$ and $b_{j+1}$
(for this graph, the 1-planar drawing on the figure is not unique).

On the one
hand, $L(A)$ has $2N+8$ vertices, including the special vertices $c_1,
\ldots, c_4$, and the dotted edge $a_ia'_{i}$ exists if and only if
$i\in A$. On the other hand, $R(B)$ has $10N+43$ vertices, including
the same four special vertices
$c_1, \ldots, c_4$, and the dotted edge $b_ib_{i+1}$ exists if and
only if $i\in B$. If $A$ and $B$ are disjoint then the graph $g(L(A),
R(B))$ is clearly 1-planar. We now prove that the converse also
holds. Consider some $i\in
A\cap B$ and assume for the sake of contradiction that $g(L(A),
R(B))$ is 1-planar. Then the edge $a_ia'_i$ must cross two edges: $b_ib_{i+1}$, and
one edge incident to the degree-2 common neighbor of $b_i$ and
$b_{i+1}$, which is a contradiction. Hence this graph is $1$-planar if and only if $A\cap
B=\emptyset$.
This proves Condition~\ref{item:no-intersection} of
disjointness-expressing. Condition~\ref{item:total-size}
is satisfied with $\formerepsilon=1$ because $g(L(A), R(B))$ has order
$12N+47$. Finally regarding Condition~\ref{item:indep-size},
by subgraph induced by the closed neighborhood of $S=\{c_1,c_2,c_3,c_4\}$ in $g(L(A), R(B))$ is
independent of the choice of $A$ and $B$ and has size $s=20$.
\begin{figure}
\begin{subfigure}{0.4\linewidth}
\centering
\includegraphics[scale=.8, page=1]{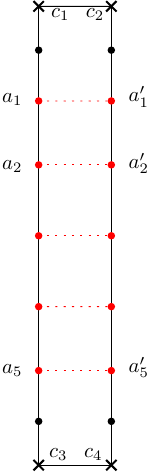}
\subcaption{$L(A)$ where $a_ia'_i\in E$ iff $i\in A$}
\label{fig: Alice part 1planar}
\end{subfigure}
\hspace{20pt}
\begin{subfigure}{0.4\linewidth}
\centering
\includegraphics[scale=.65, page=2]{1planar}
\subcaption{$R(B)$ where $b_ib_{i+1}\in E$ iff $i\in B$}
\label{fig: Bob part 1planar}
\end{subfigure}

\vspace{15pt}

\begin{subfigure}{\linewidth}
\centering
\includegraphics[scale=0.65, page=3]{1planar}
\subcaption{$g(L(A),R(B))$ obtained by identifying $c_1, \ldots, c_4$}
\label{fig:1planar-glued}
\end{subfigure}
\caption{The construction of $L$, $R$ and $g$ for 1-planar graphs in the case where $N=5$, with $A,B\subseteq \{1, \ldots, N\}$. Color red highlights edges that depend on the choice of $A$, and color blue highlights edges that depend on the choice of $B$.  }
  \label{one_planar_lowerbound_fig}
\end{figure}
\end{proof}

 It seems plausible that a generalization of the result of
 Theorem~\ref{thm:1planar} to $k$-planar graphs can be
 obtained from the same construction by adding for each edge $uv$ a
 set of $k-1$ new degree-2 vertices adjacent to both $u$ and $v$. 

 \medskip

From Theorems~\ref{thm:1planar} and \ref{thm:dex}, together with Corollary~\ref{cor:linear}, we immediately deduce the following.

\begin{theorem}
The local complexity of the class of 1-planar graphs is $\Theta(n)$.
\end{theorem}

\section{Unit-square graphs}\label{sec:usg}

Given a set $P$ of points in the plane, the \emph{unit-square graph}
associated to $P$ is the graph with vertex set $P$ in which two
points are adjacent if and only if their $\ell_\infty$-distance is at most
1. We say that a graph is a unit-square graph if it is the unit-square
graph associated to some set of points in the plane. Equivalently,
unit-square graphs can be defined as follows: the vertices correspond to
axis-parallel squares of side length 1 (or any fixed value $r$, the same for all
squares), and two vertices are adjacent if and only if the corresponding
squares intersect. The equivalence can be seen by associating to each
square its center, and to each point the axis-parallel square of side
1 centered in this point. It will sometimes be convenient to consider the two (equivalent)
definitions at once, as each of them has some useful properties.

\medskip

We say that a unit-square graph $G$ is \emph{embedded in the plane} if it
is given by a fixed set $P$ of points as above (or equivalently a set
of unit-squares). The embedding is then referred to as a
\emph{unit-square embedding} of $G$.
Recall that by Proposition~\ref{pro:trifree}, any triangle-free
unit-square graph $G$ embedded in the plane gives rise to a planar
graph embedding of $G$. Note that some planar graph embeddings of $G$
do not correspond to any unit-square embedding.

\medskip

We start with some simple observations about triangle-free unit-square graphs.

\begin{observation}\label{obs:usg1}
Let $G$ be a triangle-free unit-square graph associated to a set
$\mathcal{D}=(D_v)_{v\in V(G)}$ of
unit-squares. Then $G$ has maximum degree 4, and for each vertex $v$
of degree 4 in $G$, each of the four corners of $D_v$ is contained in the
square of a different neighbor of $v$ (and $D_v$ contains the opposite
corner of each of the squares of the neighbors of $v$).
\end{observation}

Given a triangle-free unit-square graph $G$ associated to a set
$\mathcal{D}=(D_v)_{v\in V(G)}$ of
unit-squares, and a vertex $v$, we denote by $n_{00}(v)$, $n_{10}(v)$, $n_{01}(v)$ and $n_{11}(v)$ the
neighbors of $v$ whose squares intersect the bottom-left,
bottom-right, top-left, and top-right corner of the square of $v$,
respectively (see Figure~\ref{fig:n00}). In
general these vertices might coincide, but since $G$ is triangle-free
there is at most one neighbor of each type.
By Observation~\ref{obs:usg1}, when $v$ has degree 4, all vertices
$n_{00}(v)$, $n_{10}(v)$, $n_{01}(v)$ and $n_{11}(v)$ exist and are
distinct.

\begin{figure}[htb]
\centering
\includegraphics[scale=0.4]{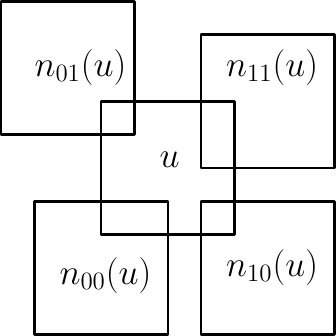}
\caption{A vertex of degree 4 in a triangle-free unit-square graph.}
\label{fig:n00}
\end{figure}

\medskip

For each unit-square graph $G$ embedded in the plane and each
vertex $u$ in $G$, we denote by $x(u)$ and $y(u)$ the $x$- and
$y$-coordinates of the center of the square associated to $u$ in the embedding.

\begin{observation}\label{obs:usg3}
Let $G$ be a triangle-free unit-square graph embedded in the plane,
and let $u,v,w$ be three distinct vertices.\begin{itemize}
\item If $v=n_{11}(u)$ and $w=n_{10}(u)$, then $y(v)> y(w)+1$;
\item If $v=n_{01}(u)$ and $w=n_{00}(u)$, then $y(v)> y(w)+1$;
\item If $v=n_{10}(u)$ and $w=n_{00}(u)$, then $x(v)> x(w)+1$;
\item If $v=n_{11}(u)$ and $w=n_{01}(u)$, then $x(v)> x(w)+1$.
\end{itemize}
\end{observation}

\begin{figure}[htb]
\centering
\includegraphics[scale=1.2]{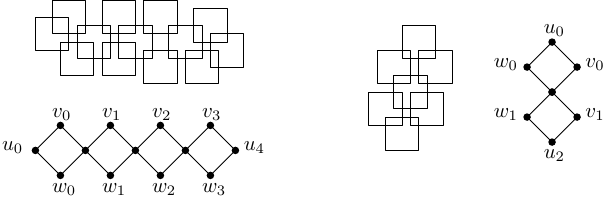}
\caption{A horizontal prop of length 4 (left), a vertical prop of
  length 2 (right),
  and the associated planar graph embeddings (which are unique up to
  reflection, by Observation~\ref{obs:usg5}). }
\label{fig:prop}
\end{figure}

Let $G$ be a triangle-free unit-square graph embedded in the plane.
For $n\ge 1$, a \emph{horizontal prop} of length $n$ in $G$ is a sequence of
distinct vertices $(u_i)_{0\le i\le n}$, $(v_i)_{0\le i\le n-1}$, and
$(w_i)_{0\le i\le n-1}$, such that the following
holds: for each $0\le i \le n-1$, $v_i=n_{11}(u_i)$,
$w_i=n_{10}(u_i)$, and $u_{i+1}=n_{10}(v_i)=n_{11}(w_i)$. 
Similarly, a \emph{vertical prop} of length $n$ in $G$ is a sequence of
distinct vertices $(u_i)_{0\le i\le n}$, $(v_i)_{0\le i\le n-1}$, and
$(w_i)_{0\le i\le n-1}$, such that the following
holds: for each $0\le i \le n-1$, $v_i=n_{10}(u_i)$,
$w_i=n_{00}(u_i)$, and $u_{i+1}=n_{00}(v_i)=n_{10}(w_i)$ (see Figure~\ref{fig:prop}). 
The vertices $u_0$ and $u_n$ are respectively said to be the starting and ending vertex of the
(horizontal or vertical) prop.

\medskip

We easily deduce the following from Observation~\ref{obs:usg3}.

\begin{observation}\label{obs:usg4}
Let $G$ be a triangle-free unit-square graph embedded in the plane,
and let $u_0$ and $u_n$ be the starting and ending vertices of a prop of length $n$ in $G$. If the
prop is horizontal, then $x(u_n)> x(u_0)+n$. If the prop is vertical, then $y(u_0)> y(u_n) + n$.
\end{observation}

Let $\mathrm{Pr}_n$ denote the graph induced by a (vertical or horizontal) prop of
length $n$. That is, $\mathrm{Pr}_n$ can be obtained from $n$ disjoint 4-cycles
$u_iv_iu_i'w_i$ ($0\le i \le n-1)$ by identifying $u_i'$ with $u_{i+1}$, 
for each $0\le i \le n-2$. 
Then $\mathrm{Pr}_n$ has $3n+1$ vertices.
Consider any fixed embedding of $\mathrm{Pr}_n$ in
the plane as a unit-square graph. Since $\mathrm{Pr}_n$ is triangle-free, this
unit-square 
embedding of $\mathrm{Pr}_n$ also gives a planar embedding of $\mathrm{Pr}_n$ (with the
same circular order of neighbors around each vertex), see Proposition~\ref{pro:trifree}. There are
multiple non-equivalent planar embeddings of $\mathrm{Pr}_n$, however a simple
area computation shows that in any planar graph drawing coming from
a unit-square embedding of $\mathrm{Pr}_n$ each 4-cycle of $\mathrm{Pr}_n$ is a face, distinct
from the outerface, so up to reflection  the resulting planar
embedding is unique. This
implies the following.

\begin{observation}\label{obs:usg5}
Let $G$ be a triangle-free unit-square graph embedded in the plane
and let $H$ be a subgraph of $G$ that is isomorphic to $\mathrm{Pr}_n$. Then $H$ is a vertical or
horizontal prop of length $n$ in $G$.
\end{observation}

Observation \ref{obs:usg5} gives a way to force some squares to align
vertically or horizontally. We will now use this to obtain a
construction with two specific vertices $v_1$ and $v_2$, and a
sequence of vertices $u_1,\ldots,u_n$ such that in any unit-square
embedding, the segment $[v_1v_2]$ is close to a diagonal and
$u_1,\ldots,u_n$ are close to be evenly spaced on this segment. We
call the latter property \emph{almost-perfect rigidity} (see the proof
below for more details).

\smallskip

We are now ready to prove the main result of this section.

\begin{theorem}\label{thm:usquare}
  The class of triangle-free unit-square graphs is $(6,
  1)$-disjointness-expressing.
\end{theorem}
\begin{proof}
Consider the graph $G_k$ depicted in Figure~\ref{fig:Gk}, where $k\ge 1$ is an 
integer. It consists of:

\begin{itemize}
\item
  a cycle $C$ of length $16k+16$, depicted with
  bold black edges in Figure~\ref{fig:Gk};
\item a copy of $\mathrm{Pr}_{8k+6}$ with endpoints $v_1, v_3\in C$ with $v_1$ and $v_3$ antipodal\footnote{Two vertices $u,v$
  in an even  cycle $C$ are said to be \emph{antipodal} if $u$ and $v$
  lie at distance $\tfrac12|V(C)|$ on $C$.} on $C$;
  \item another copy of $\mathrm{Pr}_{8k+6}$ with endpoints $v_2, v_4\in C$,
  such that $v_2$ and $v_4$ are antipodal on $C$ and, for any
  $1\le i \le 4$, $v_i$ and $v_{i+1}$ are at distance $4k+4$ on $C$
  (where indices are taken modulo 4 plus 1). Moreover the two props intersect in their middle, forming
    a star on 5 vertices (depicted in bold red edges in Figure~\ref{fig:Gk});
      \item on the subpath of $C$ of length $4k+4$ between $v_1$ and
        $v_2$, there exists a unique vertex set $I$ of size $k$ in
        which all vertices are at distance at least 4 from $v_1$ and
        $v_2$ on $C$, and any two vertices of $I$ are at
        distance at least 4 on $C$.
        For each vertex $v\in I$, we create two copies of a
        3 by 3 grid and add an edge between $v$ and a vertex of degree
        3 in each of the two grids. The $2k$ copies of the grid are denoted by $H_1,H_1',
        \ldots,H_k,H_k'$ in order from $v_1$ and $v_2$. 
      \end{itemize}

         \begin{figure}[htb]
\centering
\includegraphics[scale=0.6]{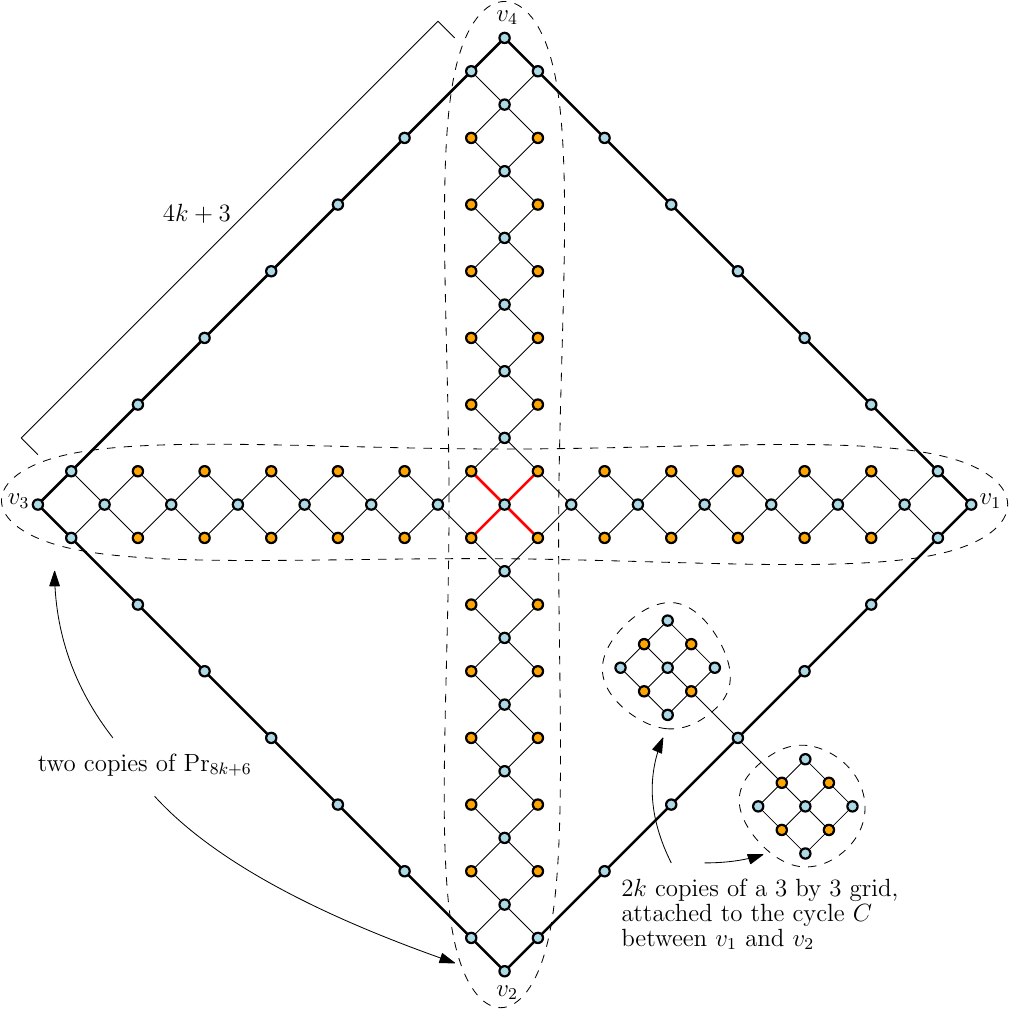}
\caption{The graph $G_k$ for $k=1$.}
\label{fig:Gk}
\end{figure}
     
      Note that $G_k$ contains exactly $(16k+16)+(48k+21)+18k=82k+37=O(k)$ vertices.      
Intuitively, up to a few technical vertex additions, two copies of $G_k$ will be used as $L(X)$ and $R(X)$ respectively, for $X\subseteq \{1, \ldots, k\}$, where $H_i$ and $H_i'$ will be removed if $i\notin X$.
Before we explain how the two copies of $G_k$ are glued together, we
analyze the properties of $G_k$.  First of all, $G_k$ can easily be realized as a
unit-square graph (see the top-left part of Figure~\ref{fig:Gk2}, where colored
vertices are represented by a square of the same color). We note that
we have only represented $G_1$ so it is not immediately clear that
several copies $H_1,H_1', \ldots,H_k,H_k'$ can fit together without
overlapping. This simply follows from the fact that consecutive copies
indexed $i,i+1$ 
of the 3 by 3 grid 
are attached to vertices lying at distance 4 on $C$.

Fix any embedding of $G_k$ as a unit-square
graph. For each vertex $v\in V(G)$, we also write $v$ for the center of
the square associated to $v$ in this embedding. Writing $d_\infty$ for
the $\ell_\infty$-distance and $d_G$ for the distance in $G$, it
follows from the definition of a unit-square graph that for any two
vertices $u,v\in V(G)$, $d_\infty(u,v)\le d_G(u,v)$. In particular
$d_\infty(v_i,v_{i+1})\le 4k+4$ and $d_\infty(v_i,v_{i+2})  \le 8k+8$ for any $1\le i \le 4$ (with indices
taken modulo 4 plus 1). By Observation~\ref{obs:usg5} and up to rotation and
reflection, we can assume that the prop with
endpoints $v_1$ and $v_3$ is a horizontal prop starting in $v_3$ and
ending in $v_1$, and that the second prop is a vertical prop
starting in $v_4$ and ending in $v_2$, precisely as in Figure~\ref{fig:Gk}.

 By Observation~\ref{obs:usg4}, $d_\infty(v_1,v_{3})  > 8k+6$ and
 $d_\infty(v_2,v_4)  > 8k+6$.
 As $d_\infty(v_i,v_{i+1})\le 4k+4$ for any $1\le i \le 4$, this
 implies (by the triangle inequality) that $d_\infty(v_i,v_{i+1})>
 4k+2$ for any such $i$ (with indices
taken modulo 4 plus 1). Hence, we have proved that $4k+2<
d_\infty(v_1,v_2)\le 4k+4$. Using Observation~\ref{obs:usg4}, this
implies that \[4k+2 <|x(v_1)-x(v_2)|\le 4k+4 \text{ and } 4k+2
<|y(v_1)-y(v_2)|\le 4k+4,\] where $x(\cdot)$ and $y(\cdot)$ denote the
$x$- and $y$-coordinates of the points as before.
    Let us denote by $u_1, u_2, \ldots,
u_{4k+5}$ the vertices on the subpath $P$ of $C$ of length $4k+4$
between $v_1$ and $v_2$ (with $u_1=v_1$ and $u_{4k+5}=v_2$). The
previous result implies that for any $1\le i \le j \le 4k+5$,
\[j-i-2<
|x(u_i)-x(u_j)|\le j-i \text{ and } j-i-2<
|y(u_i)-y(u_j)|\le j-i.\] In particular, if we translate and rotate
the vertices of the embedding of $G_k$ such that $v_1$ lies at coordinates $(0,0)$ and $v_2$
lies at $\ell_\infty$-distance at most 2 from $(4k+4,4k+4)$, then for
any $1\le i  \le 4k+5$, $u_i$ lies in the square with corners
$(i-3,i-3)$ and $(i-1,i-1)$. We call this property \emph{the
  almost-perfect rigidity} of $G_k$.

The paragraph above also shows that if we only consider the subgraph
of $G_k$ induced by the vertices
of $C$ and the vertices of the two props, in any unit-square
embedding of this graph the outerface of the embedding must be bounded
by $C$ (otherwise one of the subpaths dividing $C$ would have to be
much longer than what is possible). By Observation~\ref{obs:usg1}, for
each vertex $v\in I$, one of the copies of the 3 by 3 grid attached to
$v$ lies inside
$C$ while the other must lie outside $C$ (it can be checked that the
two copies cannot intersect two adjacent corners of a given square,
since otherwise the two copies would overlap).
In conclusion, the planar graph embedding corresponding to a
unit-square embedding of $G_k$ is unique (up to  reflection), and
corresponds precisely to the planar embedding depicted in  Figure~\ref{fig:Gk}.

\medskip

\begin{figure}[htb]
\centering
\includegraphics[scale=0.8]{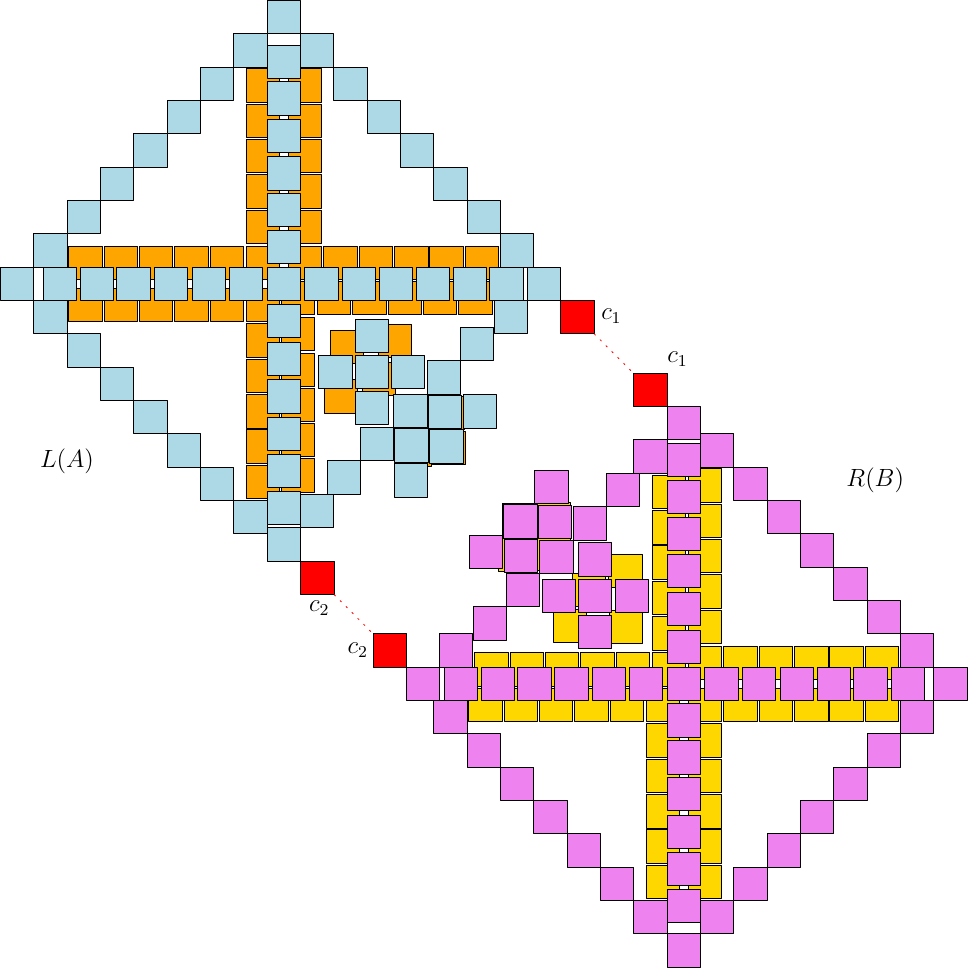}
\caption{$L(A)$ and $R(B)$ in the proof of Theorem~\ref{thm:usquare},
  when $k=1$ and $A=B=\{1\}$.}
\label{fig:Gk2}
\end{figure}

We now define $G_k'$ as the graph obtained from $G_k$ by adding a
vertex $c_1$ adjacent to $v_1$ and a vertex $c_2$ adjacent to
$v_2$, and setting $S=\{c_1,c_2\}$ as a set of special vertices. For every $X\subseteq \{1,\ldots,N\}$, we set $L(X)$ and $R(X)$
as $G_N'$, in which we delete all copies $H_i$ and $H_i'$ ($1\le i \le N$) such
that $i\not\in X$. For $A,B\subseteq \{1,\ldots,N\}$, $g(L(A),R(B))$
is obtained from $L(A)$ and $R(B)$ by gluing them along their special
vertices. This is illustrated in  Figure~\ref{fig:Gk2}, where $L(A)$ and $R(B)$ have
been drawn disjointly, for the sake of clarity, and in  Figure~\ref{fig:Gk3},
where only the interface between $L(A)$ and $R(B)$ was
represented. As illustrated in Figure~\ref{fig:Gk3} (left), it follows
from the almost-perfect rigidity of $G_k$ that if $H_i$ and $H_i'$ are present both in $L(A)$ and
$R(B)$, then some square of $H_i$ or $H_i'$ in  $L(A)$ must intersect
some square of $H_i$ or $H_i'$ in $R(B)$, which is a contradiction as
there are no edges between these copies in $g(L(A),R(B))$.

The results
obtained above show that for any $A,B\subseteq
\{1,\ldots,N\}$, $g(L(A),R(B))$ is triangle-free, and it a unit-square graph if and only if
$A$ and $B$ are disjoint. As $G_N'$ and $G_N$ have $O(N)$ vertices and
the subgraph induced by the closed neighborhood of $S$ is independent of $A$ and $B$ and contains
6 vertices, the class of triangle-free unit-square graphs is $(6, 1)$-disjointness-expressing.
\end{proof}

\begin{figure}[htb]
\centering
\includegraphics[scale=0.8]{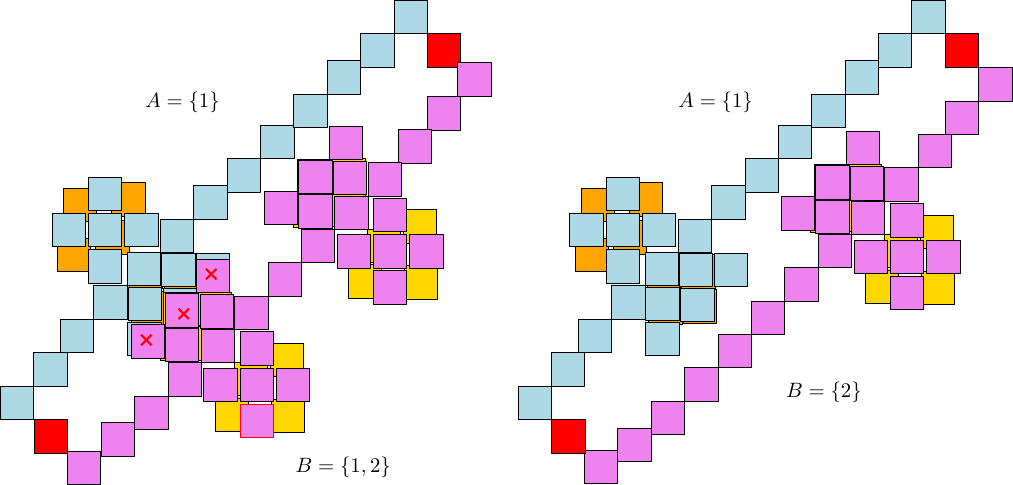}
\caption{The interface between $L(A)$ and $R(B)$ when $k=2$ and $A=\{1\}$ and
  $B=\{1,2\}$ (left), and when $A=\{1\}$ and
  $B=\{2\}$ (right). The red crosses indicate the squares of $H_i$ or $H_i'$ in $L(A)$ that intersect those of $H_i$ or $H_i'$ in $R(B)$.}
\label{fig:Gk3}
\end{figure}

In the proof of Theorem~\ref{thm:usquare} we have shown that when $A$
and $B$ are not disjoint, then the resulting graph $g(L(A),R(B))$ is not a
unit-square graph. Using Remark~\ref{rmk:disjointness implied by other class}, we obtain the following as a
direct consequence.

\begin{corollary}\label{cor:usquare}
  The class of unit-square graphs is $(6,
  1)$-disjointness-expressing.
\end{corollary}

Using Theorem~\ref{thm:dex}, together with Corollaries~\ref{cor:linear} and \ref{cor:nlogn}, we immediately deduce the following.

\begin{theorem}
  The local complexity of the class of triangle-free unit-square graphs
is $\Theta(n)$, and the local complexity of the class of unit-square
graphs is $\Omega(n)$ and $O(n\log n)$.
\end{theorem}

 We note that the proof approach of Theorem~\ref{thm:usquare} naturally extends to
higher dimension.

\section{Unit-disk graphs}\label{sec:udg}

\subsection{Definition}

The Euclidean distance between two points $x$ and $y$ in the plane is
denoted by $d_2(x,y)$, to avoid any confusion with the
$\ell_\infty$-distance $d_\infty$ considered in the previous section, and the
distance $d_G(u,v)$ between two vertices $u$ and $v$ in a graph $G$.
Given a set $P$ of points in the plane, the \emph{unit-disk graph}
associated to $P$  is the graph with vertex set $P$ in which two
points are adjacent if and only if their Euclidean distance is at most
1. We say that a graph is a unit-disk graph if it is the unit-disk
graph associated to some set of points in the plane. Equivalently,
unit-disk graphs can be defined as follows: the vertices correspond to
disks of radius $\tfrac12$ (or any fixed radius $r$, the same for all
disks), and two vertices are adjacent if and only if the corresponding
disks intersect. In particular, penny graphs are unit-disk graphs.

\subsection{Discussion}\label{sec:outudg}

We would like to prove a variant of Theorem~\ref{thm:usquare} for
unit-disk graphs, but there are two major obstacles. The first is that
there does not seem to be a simple equivalent of a horizontal or vertical prop in unit-disk
graphs, that is a unit-disk graph with $O(n)$ vertices with two specified
vertices that are at Euclidean distance at least $n$ in any unit-disk
embedding. Our construction of such a graph will be significantly more
involved. The second obstacle comes
from Pythagoras' theorem: In the unit-square case, if we consider a
path $P$ of length $n+O(1)$ between two vertices $u,v$ embedded in the
plane such  that their $x$- and $y$-coordinates both differ by exactly $n$, then in any unit-square embedding of $P$, the
vertices of $P $ deviate by at most a constant from the line segment $[u,v]$
between $u$ and $v$. This is what we used in the proof of the previous
section to make sure
that $L(A)$ and $R(B)$ (intuitively, the left and right parts) are so close that the $i$-th gadget cannot exist both on $L(A)$ and $R(B)$ simultaneously when $i\in A\cap B$ 
(i.e., subsets $A$ and $B$ must be disjoint for the graph to be in the class).
However, Pythagoras' theorem implies that
in the Euclidean case, when the Euclidean distance between $u$ and $v$
is equal to $n$, the vertices of $P$ can deviate by
$\Theta(\sqrt{n})$ from the line segment $[u,v]$, which is
too much for our purpose (we need a constant deviation). So we need
different ideas to make sure the gadgets 
are embedded sufficiently close to each other.

\medskip

\begin{figure}[htb]
\centering
\includegraphics[scale=0.8]{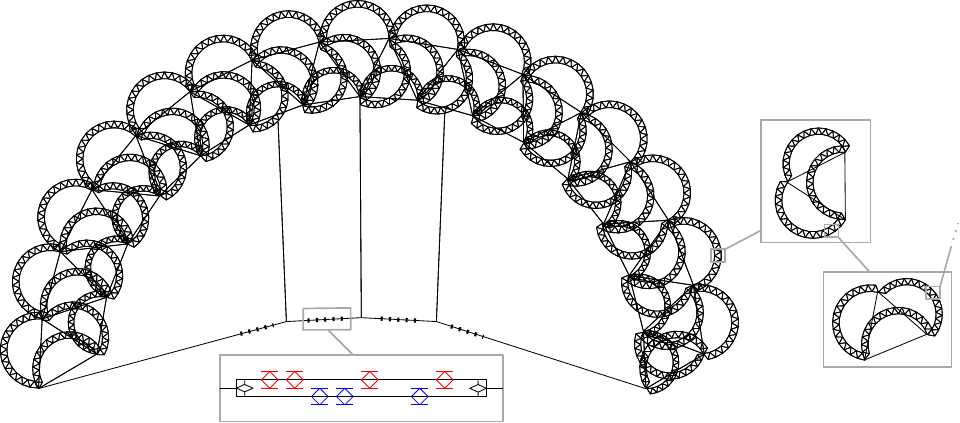}
\caption{A summary of the construction used in the proof of Theorem~\ref{thm:udg}.}
\label{fig:resume}
\end{figure}

A summary of our construction is depicted in Figure~\ref{fig:resume}. We now describe the different steps of the
construction in details. In order to not disrupt the flow of reading, the technical proofs from
this section have been deferred to the appendix (with the same
numbering). This is marked by a symbol \apx{sec:prstri} which links to
the relevant appendix section.

\subsection{Stripes}\label{sec:stri}

In this section we discuss the properties of triangles and structures
called stripes, that will be used later to describe the behavior of
unit disks.
For $k>0$ and $\delta\in [0,k]$,  a triangle in the plane is said to be
\emph{$(k,\delta)$-almost-equilateral} if all sides have length at least
$k-\delta$ and at most $k+\delta$. By the law of cosines
and the approximation $\arccos(1/2-x)=\pi/3+O(x)$ as $x\to 0$, we have
the following.

\begin{observation}\label{obs:angle}
All angles in a $(k,\delta)$-almost-equilateral triangle with $k \gg \delta$ are between
$\tfrac{\pi}3-O(\delta/k)$ and $\tfrac{\pi}3+O(\delta/k)$.
\end{observation}

For $\ell\ge 0$, the \emph{stripe} $S_\ell$ with vertex set
$u_0,v_0,\ldots,u_\ell,v_\ell$ is the graph defined as follows:
$u_0v_0$ is an edge and for any $i\ge 1$, $u_i$ is adjacent to
$u_{i-1}$ and $v_{i-1}$, and $v_i$ is adjacent to $v_{i-1}$ and
$u_i$. Note that this graph can also be obtained from a sequence of
triangles by gluing any two consecutive triangles on one of their
edges. The vertices $u_0$ and $v_\ell$ are called the \emph{ends} of
the stripe $S_\ell$.

We say that the stripe $S_\ell$ has a \emph{$(k,\delta)$-almost-equilateral embedding}
in the plane if the vertices $u_0,v_0,\ldots,u_\ell,v_\ell$ are
embedded in the plane in such way that all triangles of $S_\ell$ are
$(k,\delta)$-almost-equilateral (see Figure~\ref{fig:stripe1} for an
illustration where $k \gg \delta$).

          \begin{figure}[htb]
\centering
\includegraphics[scale=0.6]{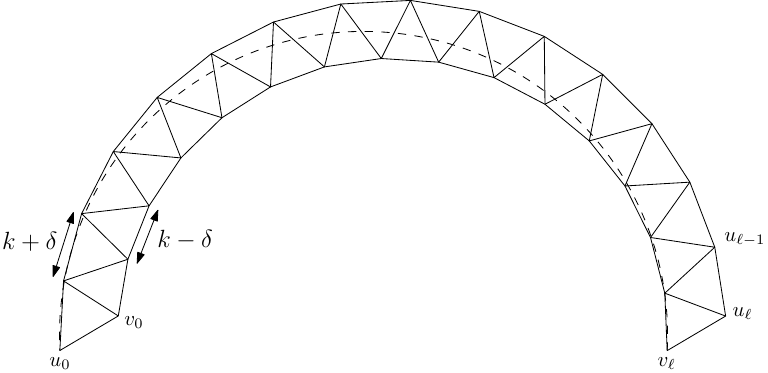}
\caption{A $(k,\delta)$-almost-equilateral embedding of a stripe, where
  the Euclidean distance between the ends is minimized.}
\label{fig:stripe1}
\end{figure}

\medskip

We show that in a $(k,\delta)$-almost-equilateral embedding
of the stripe that minimizes the Euclidean distance between its ends,
the vertices of the stripe are close to a circular arc whose radius
only depends on  $(k,\delta)$. The \emph{Menger curvature} of a triple of points $a,b,c$ is the
    reciprocal of the radius of the circle that passes through $a$,
    $b$, and $c$. 

\begin{restatable}[\apx{sec:prstri}]{lemma}{lemcurv}\label{lem:curv}
For every $k$ and $\delta=o(k)$, there are  $\rho=\rho(k,\delta)$ and $\rho'=\rho'(k,\delta)$
such
that for any $\ell$, the following holds. Consider a $(k,\delta)$-almost-equilateral embedding
of the stripe $S_\ell$ that minimizes $d_2(u_0,v_\ell)$. Then (up to
changing all $u_i$'s by $v_i$'s and vice versa), all  triples
$v_{i-1},v_i,v_{i+1}$ have the same Menger curvature $1/\rho$, and all
triples $u_{i-1},u_i,u_{i+1}$ have the same Menger curvature
$1/\rho'$. In particular all vertices $v_i$ lie on some circular arc
of radius $\rho$, and all vertices $u_i$ lie on some circular arc
of radius $\rho'$. 
\end{restatable}

Assume that $\delta>0$. By Observation~\ref{obs:angle}, the maximum angle between the lines $u_{i}v_i$
and $u_{i+1}v_{i+1}$ is of order $O(\delta/k)$. 
Hence, there exists a
constant $\alpha$ (independent of $k$ and $\delta$) such that if
$\ell=\alpha k/\delta$, the maximum angle between
$u_0v_0$ and $u_{\ell}v_\ell$ in any $(k,\delta)$-almost-equilateral embedding is close to $\pi$.
In this case, by Lemma~\ref{lem:curv} the
minimum (Euclidean) distance $m$ between $u_0$ and $v_\ell$ in any $(k,\delta)$-almost-equilateral embedding is of order
$\Theta(k/\delta \cdot k)=\Theta(k^2/\delta)$. Moreover, any $(k,\delta)$-almost-equilateral embedding of
$S_{\ell}$ realizing this minimum $m$ is close to some semicircle
with
endpoints $u_0$ and $v_\ell$, in the sense that all the vertices of
$S_\ell$ lie at distance $O(k)$ from the semicircle (see Figure~\ref{fig:stripe1}). We will need a looser version of this observation in the
slightly weaker setting where $d_2(u_0,v_\ell)\le m+1$, instead of
$d_2(u_0,v_\ell)=m$.

\begin{restatable}[\apx{sec:prstri}]{lemma}{lemcurvv}\label{lem:curv2}
Let $\delta>0$, $k\gg\delta$ and $\ell=\lceil \alpha k/\delta\rceil$ as above, and let $m= \Theta(k^2/\delta)$ be the minimum Euclidean distance between $u_0$ and $v_\ell$
in any $(k,\delta)$-almost-equilateral embedding of $S_\ell$. Consider a $(k,\delta)$-almost-equilateral embedding of $S_\ell$ where
$d_2(u_0,v_\ell)\le m+1$. Let $c$ be the midpoint of the segment
$[u_0,v_\ell]$. Then no vertex of the stripe $S_\ell$ is contained in
the disk of center $c$ and radius $m/2-O(k)$, and all vertices of the
stripe $S_\ell$ are contained in a disk of center $c$ and  radius  $O(m)=O(k^2/\delta)$.
\end{restatable}

\subsection{Quasi-rigid graphs}\label{sec:qrg}

Our next goal is to construct a sequence of unit-disk graphs
$T_n$ on $O(n^2)$ vertices, with two vertices
$u$ and $v$ such that in any unit-disk representation of $T_n$,
$d_2(u,v)=\Omega(n)$.

We define $T_n$ as follows. We consider two adjacent vertices $u,v$ of
the infinite square grid graph and define $X$ as the set of vertices at
distance at most $2n+1$ from $u$ or $v$ in the square grid graph, and $Y\subseteq
X$ as the set of vertices at distance exactly $2n+1$ from $u$ or
$v$. Note that $|X|=2(2n+2)^2=8(n+1)^2$ and $|Y|=8n+6$. The vertices of $Y$ are denoted by $y_1,\ldots,y_{8n+6}$. The graph
$T_n$ is obtained from the subgraph of the square grid graph induced by $X$
by adding, for each vertex $y_i$ of $Y\subseteq
X$, a vertex $c_i$ adjacent
to $y_i$. We finally add edges to form a cycle $C$ containing all vertices
$c_i$ in order, together with 10 new vertices in the corners (2 or 3 at each corner, see
Figure~\ref{fig:tn}). In Figure~\ref{fig:tn}, the vertices of $C$ are
depicted in black, the vertices of $X$ are blue or red, and the
vertices of $Y$ are the blue or red vertices that are
adjacent to a vertex of $C$.

The  cycle $C$ has length $8n+16$. Note that the
resulting graph $T_n$ contains $8n^2+24n+24=O(n^2)$ vertices and is a planar triangle-free unit-disk graph (see
Figure~\ref{fig:tn}). A
simple area computation shows that in any unit-disk embedding of $T_n$,
$C$ bounds the outerface of the corresponding planar graph
embedding (by the arguments of Section~\ref{sec:gc} there is a unique
planar graph embedding in the sphere, and as all faces except $C$ have
size at most 8, $C$ must be the outerface).

          \begin{figure}[htb]
\centering
\includegraphics[scale=0.6]{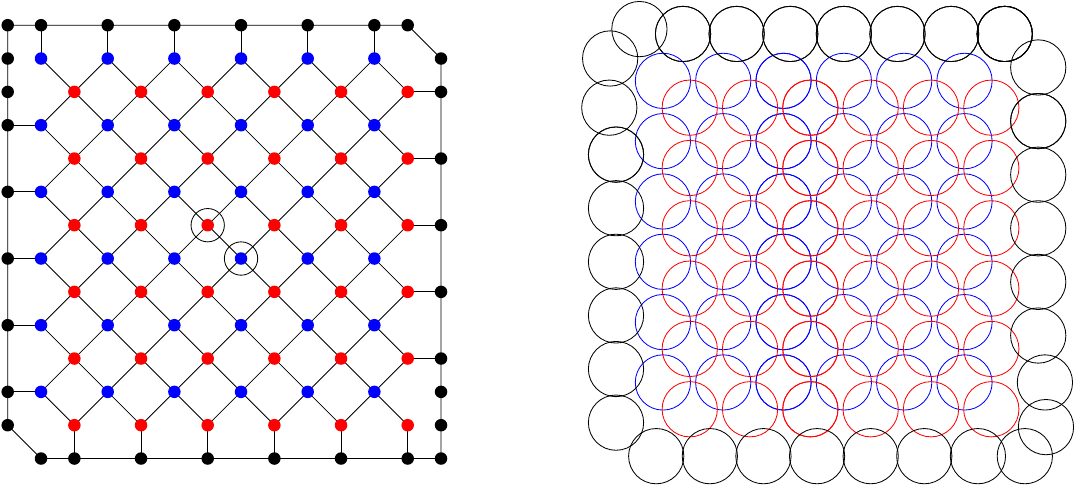}
\caption{The graph $T_n$ with $n=2$ (left), with the central vertices $u$ and $v$ circled ; and a unit-disk embedding of
  $T_n$ (right).}
\label{fig:tn}
\end{figure}

\begin{restatable}[\apx{sec:prqrg}]{lemma}{lemtk}\label{lem:tk}
Let $c$ and
$c'$ be two antipodal vertices on the cycle $C$ in $T_n$. In
any unit-disk embedding of $T_n$, $d_2(c,c')\ge \left (\pi\sqrt{2}-4 \right ) n-O(\sqrt{n})=\Omega(n)$.
\end{restatable}

Let $f(n)$ be the infimum Euclidean distance between two antipodal
vertices $c$ and $c'$ of the cycle $C$ 
in a unit-disk embedding of $T_n$. By Lemma~\ref{lem:tk},
$f(n)=\Omega(n)$. Since $C$ has length $O(n)$, it follows that
$f(n)=\Theta(n)$. Assume for simplicity that the infimum $f(n)$ is a minimum
(otherwise we work with a sequence of unit-disk embeddings such
that the Euclidean distance between $c$ and $c'$ tends to $f(n)$).
Let $Z_n$ be the point set of a unit-disk embedding of $T_n$ in
which $d_2(c,c')=f(n)$. Add $\lceil f(n)\rceil -1$ points, evenly
spaced on the line-segment $[c,c']$ with respect to the Euclidean distance (note that together with $c$
and $c'$, any two consecutive
points on $[c,c']$ lie at Euclidean distance at most 1 apart).
Let $Z_n'$
denote the resulting point set and $T_n'$ be
the resulting unit-disk graph. Note that $T_n'$ has $O(n^2)$ vertices
and in any unit-disk graph embedding of $T_n'$, $f(n)\le d_2(c,c')\le
f(n)+1$, with $f(n)=\Omega(n)$.

\medskip

We say that an infinite family of unit-disk graphs $\mathcal{T}$ is
\emph{quasi-rigid with density $g$} if for arbitrarily large $k$ there
is a graph $G\in \mathcal{T}$ with at most $g(k)$ vertices, and two specific vertices
$x$ and $y$ in $G$ such
that in any unit-disk embedding of $G$, $k\le d_2(x,y)\le
k+1$. Using this terminology, the family $\mathcal{T}^2=(T_n')_{n\ge 1}$ is
quasi-rigid with density $O(k^2)$. We will now see how to construct
increasingly sparser
quasi-rigid classes, starting with  $\mathcal{T}^2$ and using stripes.

\medskip

        \begin{figure}[htb]
\centering
\includegraphics[scale=0.8]{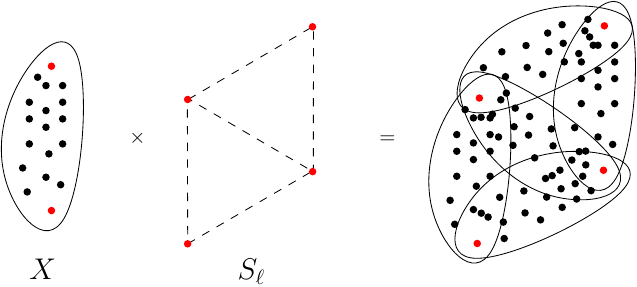}
\caption{Copies of a point set are glued along some stripe.}
\label{fig:stripe2}
\end{figure}

Assume we have found a family $\mathcal{T}^{1+c}$ which is
quasi-rigid with density $g(k)=O(k^{1+c})$, for some $c\ge 0$. Take some
graph $G \in \mathcal{T}^{1+c}$ with $g(k)$ vertices and that has two vertices $x$ and $y$
such that in any unit-disk embedding of $G$, $k\le d_2(x,y)\le
k+1$. Let $X$ be the point set associated to some unit-disk embedding of $G$, and let
$d:=d_2(x,y)$ in this point set ($k\le d \le k+1$). Consider a $(d,0)$-almost equilateral embedding of a  stripe
$S_\ell$ (for some integer $\ell\ge 0$ whose value will be fixed later), and for each edge $uv$ in $S_\ell$, add an isometric copy of
$X$ in which $x$ is identified with $u$ and $y$ is identified with
$v$ (by definition of $X$ we can always find such an isometric copy of
$X$). See Figure~\ref{fig:stripe2} for an illustration. We denote by $X^{(\ell)}$ the resulting point set, and by
$G^{(\ell)}$ the resulting unit-disk graph. Note that the different
copies of $X$ interact and thus $G^{(\ell)}$ contains more edges than
the unions of copies of $G$. Fix any unit-disk embedding of
$G^{(\ell)}$, and consider the points corresponding to $x$ and $y$ in
each of  the copies of $G$ in $G^{(\ell)}$. 
By the properties of $G$, the union of all these points form a $(d,1)$-almost equilateral embedding of a  stripe
$S_\ell$. Let $a$ and $b$ be the ends of this stripe. By Lemma~\ref{lem:curv2}, there exists a constant
$\alpha>0$ such that if we set $\ell=\lceil \alpha d\rceil=O(k)$, the following holds. Let $m$ be
the minimum distance between $a$ and $b$ in any unit-disk
embedding of  $G^{(\ell)}$. Then
$m=\Theta(d^2)=\Theta(k^2)$. As before we consider a point set realizing this
distance and we add $\lceil m\rceil -1$ evenly spaced points on the
segment $[a,b]$. We denote the resulting unit-disk graph by $H$ and we
note that in any unit-disk graph embedding of $H$, $m\le
d_2(a,b)\le m+1$. The point set $X^{(\ell)}$ has size at most $5\ell
g(k)=O(k^{2+c})$ as there are $4\ell+1$ edges in $S_\ell$.
We have added at most $m=O(d^2)=O(k^2)$ vertices to the graph, so $H$
has $O(k^{2+c})$ vertices. As $m=\Omega(k^2)$, the family of all
graphs $H$ created in this way is quasi-rigid with density
$g'(k)=O(k^{1+c/2})$.

\medskip

By iterating this construction, starting with $\mathcal{T}^2$ and using 
Lemma~\ref{lem:curv2}, we obtain the following result. See
Figure~\ref{fig:resume} for a depiction of the construction.

\begin{observation}\label{obs:cclscha}
  For any $\varepsilon>0$ there is a
family of unit-disk graphs which is quasi-rigid with density
$O(k^{1+\varepsilon})$.
More precisely,
  for any sufficiently large $k$ there is a unit-disk
graph $G$ with $O(k^{1+\varepsilon})$ vertices with two
vertices $a$ and $b$ such that in any unit-disk
embedding of $G$, $k\le d_2(a,b)\le k+1$, and $a$ and $b$ are
joined by a path $P$ of length at most $k+1$ such that at least $k/2$
consecutive vertices of $P$ lie at Euclidean distance at least
$\Omega(k)$ from $G-P$.
\end{observation}

\subsection{Tied-arch bridges}\label{sec:tab}

In the previous subsection we have constructed unit-disk graphs with a pair of
vertices $a,b$ whose possible Euclidean distance in any unit-disk
embedding lies in some interval $\left[m,m+1\right]$. This is still too much for our purposes, because Pythagoras'
theorem then implies that a path $P$ of length $\lceil m\rceil$ between
$a$ and $b$ might deviate from the line-segment $[a,b]$ by
$\Omega(m^{1/2})$, which prevents us from using arguments similar to
unit-square case. We now explain how to obtain an even tighter path. The
idea will be to cut $P$ into $\log m$ subpaths of nearly equal length,
and join the endpoints of these subpaths to the rest of the graph,
using some paths of minimum length. See Figure~\ref{fig:sketch} for an
illustration of this step of the proof. We will then argue that for any
unit-disk embedding, at least one of these subpaths will be maximally tight (i.e., at constant distance from the line-segment
joining its endpoints). 

\medskip

       \begin{figure}[htb]
\centering
\includegraphics[scale=1]{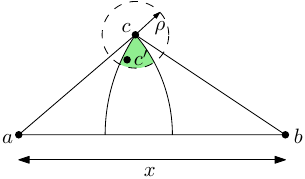}
\caption{The setting of Lemma~\ref{lem:midtri}.}
\label{fig:triangle}
\end{figure}

In this section, $x>0$ and $0\le \delta\le 1$ are real numbers, and
whenever we use the
$O(\cdot)$ notation, we implicitly assume that $x\to \infty$ (see for
instance the terms $O(1/x)$ and $O(\sqrt{\delta/x})$ in the statement
of the next lemma).
Let $abc$ be a triangle such that
$d_2(a,b)= x$ and $d_2(a,c)+d_2(b,c)= x+\delta$, with
$|d_2(a,c)-d_2(b,c)|\le 1$.  Assume by symmetry that $d_2(a,c)\le
d_2(b,c)$. In particular $\tfrac12(x+\delta -1)\le d_2(a,c)\le
\tfrac12(x+\delta )$ and  $\tfrac12(x+\delta )\le d_2(b,c)\le
\tfrac12(x+\delta +1)$. Let $c'$ be
a point such that $d_2(a,c')\le d_2(a,c)$, $d_2(b,c')\le d_2(b,c)$,
and $d_2(c,c')\le \rho$, for some constant $\rho=O(1)$. See Figure~\ref{fig:triangle} for an illustration.

\begin{restatable}[\apx{sec:prtab}]{lemma}{lemmidtri}
\label{lem:midtri}
  $\min\{d_2(a,c)- d_2(a,c'), d_2(b,c)- d_2(b,c')\}\le 
 \rho \sqrt{2\delta/x}+O(1/x)=O(\sqrt{\delta/x})$.
\end{restatable}

Consider a graph $G_0$ given by Observation~\ref{obs:cclscha} with
parameter $2k$. The
unit-disk graph $G_0$ thus contains $O(k^{1+\varepsilon})$ vertices and has two
vertices $u'$ and $v'$ such that in any unit-disk  embedding
$2k\le d_2(u',v')\le 2k+1$ and $u'$ and $v'$ are joined by a path $P'$ of
length less than $2k+1$.

\medskip

For any unit-disk embedding of a unit-disk graph $G$ and for any  $\delta>0$, we say that a path with endpoints $a$ and $b$ in $G$ 
is \emph{$\delta$-tight} if the length of the path and the
Euclidean distance between $a$ and $b$ differ by at most
$\delta$. With this terminology, the path $P'$ defined in the previous
paragraph is 1-tight for any unit-disk embedding of $G_0$. Note that by the triangle inequality, any
subpath of a $\delta$-tight path is also $\delta$-tight.

\medskip

By the second part of Observation~\ref{obs:cclscha}, $P'$ has a
subpath $P$ of length $k$ with
endpoints $u$ and $v$ such that all vertices
of $P$ lie at Euclidean distance at least $\Omega(k)$ from $G_0-P'$ in any
unit-disk embedding. It will be convenient to work with this subpath $P$ instead
of $P'$, as a large region around $P$ is free of any vertices of
$G_0-P'$. For any unit-disk embedding of $G_0$, since $P'$ is 1-tight, $P$ is also
1-tight and thus $k\le d_2(u,v)\le k+1$.

\medskip

It is helpful to have a quick look at Figure~\ref{fig:sketch} before
reading this section. We consider a vertex $w$ of $P$, which divides $P$ into two consecutive
paths $P_0$ and $P_1$, whose lengths differ by at most 1. We consider
a unit-disk embedding of $G_0$ and look at the perpendicular
bisector of the line-segment $[u,v]$. By connectivity, this line
intersects $G_0-P$. Let $z$ be the first vertex of $G_0-P$ whose
unit-disk is intersected
by this line (if several such vertices exist we pick one arbitrarily); by the properties above, $z$ lies at distance
at least $\Omega(k)$ from $w$. As in the construction of quasi-rigid
graphs above Observation~\ref{obs:cclscha}, we consider a point set
corresponding to a unit-disk embedding of $G_0$ where the
Euclidean distance $m^*$ between $z$ and $w$ is
minimized and add, along the segment $[zw]$, $\lceil m^* \rceil
-1$ evenly spaced new points.
Observe that in the resulting unit-disk graph $G_1$ the newly added
vertices correspond to a
path $Q$ of minimum length between $z$ and $w$ (and possibly some extra edges between this
path and $V(G_0)$).
Therefore we have $m^*\leq d_2(z,w) \leq \lceil m^* \rceil$ in any unit-disk embedding of $G_1$.
We iterate this procedure recursively on $P_0$ and $P_1$,
creating 4 consecutive subpaths $P_{00}, P_{01}, P_{10}, P_{11}$ of
$P$, and two new paths $Q_0$ and $Q_1$ joining the new midpoints to $G_0-P$. More precisely, for any $i \ge 0$, consider a unit-disk embedding of $G_{i}$,
  and any subpath $P_\mathbf{x}$ of $P$ with 
  $\mathbf{x}\in \{0,1\}^i$ in $G_i$. Let $a,b$ denote the endpoints
  of  $P_\mathbf{x}$. Then in $G_{i+1}$, $P_\mathbf{x}$ is split between $P_{\mathbf{x}0}$ (with endpoints
$a$ and $c'$) and  $P_{\mathbf{x}1}$ (with endpoints
$c'$ and $b$), where $c'$ is a midpoint of $P_\mathbf{x}$, which is then
joined (in the way described above) via some path $Q_{\mathbf{x}}$ of minimal length to some vertex $z_{\mathbf{x}}$ of
$G_i-P$ lying at
distance at most $\tfrac12$ from the perpendicular bisector of the
line segment $[ab]$. We precise that at each recursive step $\mathbf{x}\in \{0,1\}^i$ , the
paths $Q_{\mathbf{x}0}$ and $Q_{\mathbf{x}1}$  are added one after the
other: we first add a path $Q_{\mathbf{x}0}$ of minimal length (such
that the resulting graph is still a unit-disk graph), and
then we add a path $Q_{\mathbf{x}1}$ of minimum length (such
that the resulting graph is still a unit-disk graph). Note that by construction $G_{i+1}$ is a
unit-disk graph.
See Figure~\ref{fig:sketch} for a picture of
$G_3$.

\medskip

        \begin{figure}[htb]
\centering
\includegraphics[scale=0.7]{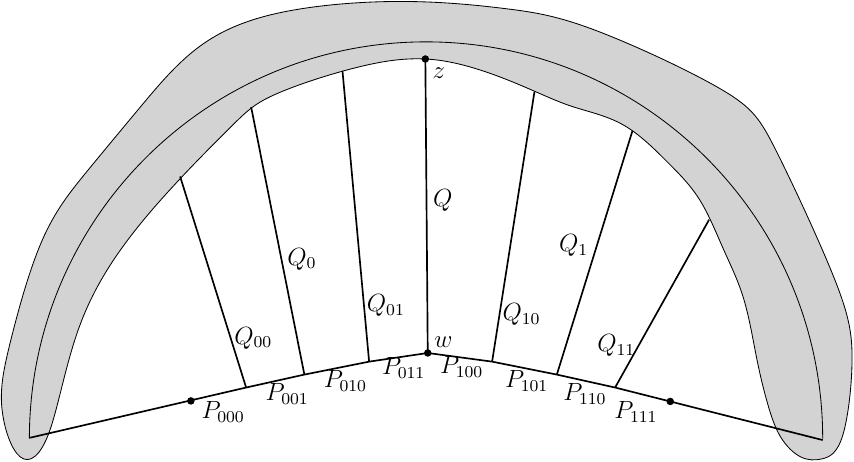}
\caption{A tied-arch bridge.}
\label{fig:sketch}
\end{figure}

Using the notation introduced in the previous paragraph, we obtain the 
following simple consequence of Lemma~\ref{lem:midtri}.

\begin{restatable}[\apx{sec:prtab}]{corollary}{cormidtri}
\label{cor:midtri}
  For any $\delta>0$ there exists
  $\gamma=O(\sqrt{\delta/d_2(a,b)})$ such that in any unit-disk
  embedding of $G_{i+1}$, if
  $P_\mathbf{x}$ is $\delta$-tight, then $P_{\mathbf{x}0}$ is
    $\gamma$-tight or $P_{\mathbf{x}1}$ is
    $\gamma$-tight.
\end{restatable}

Note that for any unit-disk embedding of $G_{i+1}$, the restriction of
the embedding to $V(G_i)\subset V(G_{i+1})$ is a unit-disk embedding
of $G_i$ (the
difference between $V(G_i)$ and $V(G_{i+1})$ being the union of the newly added
paths $Q_{\mathbf{x}}$). 
The following is proved by induction on $i\ge 0$.

\begin{restatable}[\apx{sec:prtab}]{lemma}{lemfixpath}
\label{lem:fixpath}
In any unit-disk embedding of $G_i$, there is $\mathbf{x}\in
\{0,1\}^i$ such that $P_\mathbf{x}$ is $O(2^{i}\cdot k^{2^{-i}-1})$-tight.
\end{restatable}

Consider a unit-disk embedding of a unit-disk graph $G$. For $\ell\ge
2$ and 
$\gamma\ge 0$, we say that a path $P=v_1,\ldots,v_{\ell+1}$ of
length $\ell$ in $G$ is \emph{$\gamma$-regular} if the following holds: If we place $\ell+1$ evenly
spaced points $u_1,\ldots,u_{\ell+1}$ on the line-segment $[v_1,v_{\ell+1}]$ with
$u_1=v_1$ and $u_{\ell+1}=v_{\ell+1}$, 
then for any $1\le i \le \ell+1$, $d_2(u_i,v_i)\le \gamma$.
For $s\le \ell+1$, we say that $P$ is \emph{$(\gamma,s)$-regular} if
the above holds for any $1\le i \le s$ (instead of $1\le i \le
\ell+1$, so that $(\gamma,\ell+1)$-regular is the same as
$\gamma$-regular for a path of length $\ell$) .
In other words, this
means that the first $s$ vertices of $P$ are close to their ideal
location on the segment connecting the two endpoints of $P$. 

\begin{restatable}[\apx{sec:prtab}]{lemma}{lemregpath}
\label{lem:regpath}
Consider a unit-disk embedding of a unit-disk graph $G$, and let
$P=v_1,\ldots,v_{\ell+1}$ be a path of length $\ell$ in $G$. If $P$ is
$\delta$-tight with $\delta\le 1$, then $P$ is
$\gamma$-regular with
$\gamma:=\sqrt{\ell\delta/2}+O(\ell^{-1/2})$. Moreover, for any
$\alpha>0$, $P$ is $(\lambda,\alpha\ell)$-regular with
$\lambda:=\sqrt{(2\alpha-\alpha^2)\gamma^2+\alpha^2}$ (in particular,
when $\gamma=O(1)$, $\lambda$ can be made arbitrarily small by taking
$\alpha$ arbitrarily small but constant).
\end{restatable}

Set $t:=\lceil\log \log
k\rceil $ (we recall that $\log$ denotes the logarithm in base 2). 
We immediately deduce the following. 

\begin{corollary}\label{cor:regpath}
   In any unit-disk embedding of $G_t$, there is $\mathbf{x}\in
\{0,1\}^t$ such that $P_\mathbf{x}$ is $O(k^{-1}\log k)$-tight, and in particular $O(1)$-regular.
\end{corollary}

\begin{proof}
 By Lemma~\ref{lem:fixpath}, there is $\mathbf{x}\in
\{0,1\}^t$ such that the length of the path $P_\mathbf{x}$ in $G_t$ with endpoints $a$ and
$b$ differs from $d_2(a,b)$ by $\delta:=O(2^t\cdot
k^{2^{-t}-1})$. As
$P_\mathbf{x}$ has length $\Theta(2^{-t}k)$, it follows from Lemma~\ref {lem:regpath} that $P_\mathbf{x}$ is $\gamma$-regular with
$\gamma=O\left(\sqrt{k^{2^{-t}}}\right)=O\left(k^{2^{-t-1}}\right)$. As
$k^{2^{-t}}\le 2^{\log k\cdot \log^{-1}k}= 2$, $\delta=O(k^{-1}\log
k)$ and $\gamma=O(1)$,
as desired.
\end{proof}

To summarize this subsection, for every $\varepsilon>0$ and any sufficiently large $k $ we
have constructed a unit-disk graph $G$ with $O(k^{1+\varepsilon})$
vertices which contains $\log k$ disjoint paths $P_1',P_2', \ldots$
of length $k/\log k$
with the following properties: In any unit-disk graph embedding of
$G$,
\begin{itemize}
\item at least one of the paths $P_i'$ is $O(k^{-1}\log k)$-tight and
  thus $O(1)$-regular; and
  \item each $P_i'$ contains a subpath $P_i$ of
    length at least $\tfrac{k}{2\log k}$, such that all vertices of $P_i$ lie
    at distance at least $\Omega(\tfrac{k}{\log k})$ from $G-P_i'$. Moreover,
    if $P_i'$ is $O(k^{-1}\log k)$-tight, then  $P_i$ is also
    $O(k^{-1}\log k)$-tight, and in particular $O(1)$-regular.
  \end{itemize}

  We call a graph $G$ with these properties a \emph{tied-arch bridge
    of density $O(k^{1+\varepsilon})$}.

\subsection{Corridors}\label{sec:cor}

In the previous subsection we have seen how to produce unit-disk graphs
with large $O(1)$-regular paths, that is paths  that are ``maximally''
 tight. In this  subsection we introduce the final tool needed to
prove the main result of this section: corridors. This is 
where, intuitively, we will place the gadgets depending on $A$ and $B$
in order for the vertices to locally decide whether the subsets are disjoint or not.

\medskip

The graph
depicted in Figure~\ref{fig:test} (bottom) with black vertices and edges (and grey
or black circles for the unit-disk embedding) is called an
\emph{$r$-corridor} with ends $u$ and $v$: it consists of 2 internally
vertex-disjoint paths of length $r+2$ between $u$ and $v$, say
$P=x_0,\ldots,x_{r+2}$ and $Q=y_0,\ldots,y_{r+2}$, with $u=x_0=y_0$
and $v=x_{r+2}=y_{r+2}$, together with vertices $x_1'$ (adjacent to
$x_0$ and $x_2$), $y_1'$ (adjacent to $y_0$ and $y_2$), $x_{r+1}'$
(adjacent to $x_{r}$ and $x_{r+2}$), $y_{r+1}'$
(adjacent to $y_{r}$ and $y_{r+2}$), $z_1$ (adjacent to $x_1'$ and
$y_1'$) and $z_{r+1}$ (adjacent to $x_{r+1}'$ and $y_{r+1}'$).

The paths $x_1,\ldots,x_{r+1}$ and $y_1,\ldots,y_{r+1}$ are called the
\emph{walls} of the corridor (we emphasize that the walls do not
contain $u$ and $v$).

\begin{figure}[htb]
\centering
\includegraphics[scale=0.8]{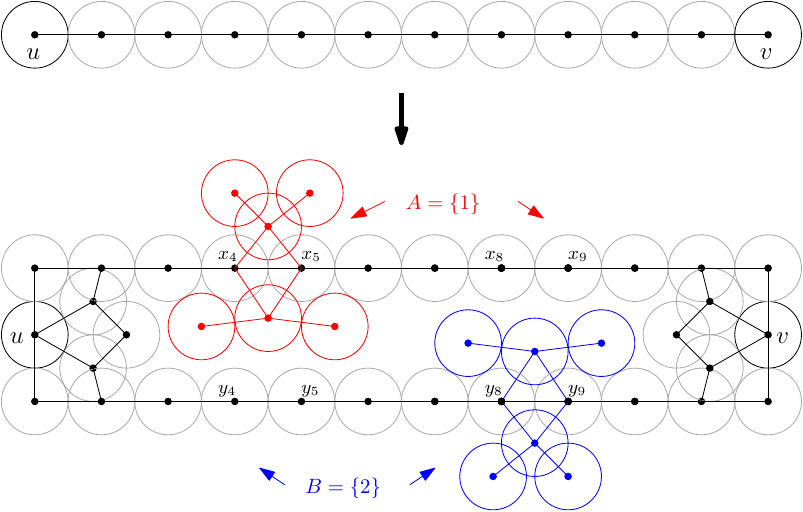}
\caption{A path between $u$ and $v$ (top) is replaced by a corridor of the
  same length between $u$ and $v$ (bottom), with sculptures decorating the walls (in red and blue). The edges of the graph show which disks intersect or not
  (even if they appear to touch on the figure).}
\label{fig:test}
\end{figure}

\begin{observation}\label{obs:corridor1}
Any $r$-corridor is a unit-disk graph, and in any unit-disk embedding
of an $r$-corridor with ends $u$ and $v$, such that $u$ and $v$ lie on
the outerface of the corresponding planar graph drawing, we have $d_2(u,v)\le r$.
\end{observation}

\begin{proof}
A unit-disk embedding of an $r$-corridor is depicted in Figure~\ref{fig:test} (bottom). The second property follows from the fact
that there is a unique planar graph embedding (up to reflection, which
follows from Theorem~\ref{thm:whi}), and the 4 neighbors of $u$ (and $v$) are non-adjacent, and thus the
angle between $ux_1$ and $uy_1$ is at least $\pi$ (and
similarly for the angle between $vx_{r+1}$ and $vy_{r+1}$).
\end{proof}

If an $r$-corridor with ends $u$ and $v$ is embedded as a unit-disk
graph in
the plane, we say that the corridor is \emph{$\delta$-tight} if $r$ and $d_2(u,v)$
differ by at most $\delta$.

\medskip

Consider a unit-disk graph $G$ embedded in the plane with a $\delta$-tight path
$P'$, for $\delta\le 1$. Assume that $P'$ has a subpath $P$ of
length $|P|\ge |P'|/2$, with
endpoints $u,v$,  such that in any
unit-disk embedding of $G$, $P$ lies at Euclidean distance at
least $4$ from $G-P'$.  As a subpath of $P'$, $P$ is also
$\delta$-tight. Let $G'$ be the graph obtained from $G$ by deleting the
internal vertices of $P$ and adding a copy of a $|P|$-corridor $C$  with
ends $u$ and $v$. We say that we have \emph{replaced the path $P$ by
  the corridor $C$ in $G$}.

\begin{observation}\label{corridorsubst}
The graph $G'$ is a unit-disk graph, and for any unit-disk graph
embedding of $G'$, there exist a unit-disk graph embedding of $G$ that
coincides with that of $G'$ on the vertex-set $\{u,v\}\cup (V(G)\setminus
V(P))=\{u,v\}\cup (V(G')\setminus V(C))$.
\end{observation}

\begin{proof}
  Let $\hat{P}$ denote the polygonal chain corresponding to the embedding
  of $P$. Let $R$ be the region of all points at Euclidean distance at
  most 2 from $\hat{P}$, and at distance more than 1 from the
  neighbors of $u$ and $v$ not in $P$. Since $P$ lies at Euclidean distance at
least $4$ from $G-P'$, $R$ is at distance more than 1 from all points
of $V(G)\setminus V(P)$. We now embed $C$ inside $R$.

 Conversely, given a unit-disk embedding of $G'$, we can simply remove
 $V(C)\setminus \{u,v\}$ and as by Observation~\ref{obs:corridor1},
 $d_2(u,v)\le |P|$, we can add a path of length $|P|$ between $u$ and
 $v$ in the region delimited by the walls of $C$ (so the newly added
 vertices do not interfere with the rest of the graph).
\end{proof}

We now consider a tied-arch bridge $G$
    of density $O(k^{1+\varepsilon})$, as
  constructed in the previous section. Recall that $G$ contains
  $\log k$ disjoint paths $P_1',P_2', \ldots$
of length $k/\log k$
with the following properties: In any unit-disk graph embedding of
$G$,
\begin{itemize}
\item at least one of the paths $P_i'$ is $O(k^{-1}\log k)$-tight and
  thus $O(1)$-regular, and
  \item Each $P_i'$ contains a subpath $P_i$ of
    length at least $\tfrac{k}{2\log k}$, such that all vertices of $P_i$ lie
    at distance at least $\Omega(\tfrac{k}{\log k})$ from $G-P_i'$. Moreover,
    if $P_i'$ is $O(k^{-1}\log k)$-tight, then  $P_i$ is also
    $O(k^{-1}\log k)$-tight, and in particular $O(1)$-regular.
  \end{itemize}

  We now replace each path $P_i$ in $G$ by a $|P_i|$-corridor $C_i$ as defined
  above. Let $G'$ be the resulting graph (and note that $G'$ still has
  $O(k^{1+\varepsilon})$ vertices).

  \begin{observation}\label{corridorsubst2}
The graph $G'$ is a unit-disk graph, and in any unit-disk graph embedding of
$G'$, some $\Theta\left (\tfrac{k}{\log k} \right )$-corridor $C_i$  is
$O(k^{-1}\log k)$-tight, and in particular the two walls of $C_i$ are
$O(k^{-1}\log k)$-tight and thus $O(1)$-regular.
\end{observation}

\begin{proof}
By Observation~\ref{corridorsubst} (applied to each corridor $C_i$),
any unit-disk embedding of $G'$ can be transformed in a unit-disk
embedding of $G$ by replacing each corridor $C_i$ by the original path
$P_i $ and leaving the rest of the embedding unchanged. In the
resulting embedding of $G$, one of the paths $P_i$ is $O(k^{-1}\log
k)$-tight. It follows that in the original embedding of $G'$, the
corridor $C_i$ is also $O(k^{-1}\log
k)$-tight, and thus so are its two walls. This implies that the two
walls of $C_i$  are $O(1)$-regular. 
\end{proof}

\subsection{Decorating the corridors}

Now that we have found a corridor in which the walls are
$O(1)$-regular, we are ready to add the gadgets that will express the disjointness. Given $A,B\subseteq \{1,\ldots,\ell\}$, one wall will belong to $L(A)$
and the other to $R(B)$.
Intuitively, our goal will be to decorate each 
wall with sculptures, the locations of which are given by the subset $A$ (resp.\ the subset $B$)
for walls belonging to $L(A)$ (resp.~$R(B)$).
The main
idea will be to make sure that a sculpture on one wall cannot be
placed next to a sculpture on the other wall because the corridor is too
narrow for that (this corresponds to the fact that $A$ and $B$ must not intersect). One important point is
that we do not know in advance which of the $O(\log k)$ corridors will
have $O(1)$-regular walls, so we have to decorate all of them in
advance (in the same way).

\medskip

Consider an $r$-corridor $C$ in some unit-disk graph $G$, and
sets $A,B\subseteq \{1,\ldots,\ell\}$, with $r\ge 4\ell+8$. Let $P_A=x_1,\ldots,x_r$ and $P_B=y_1,\ldots,y_r$
be the two walls of $C$. By \emph{decorating the walls of $C$} with $A$ and $B$, we
mean the following: for each $1\le i\le \ell$,
\begin{itemize}
\item if $i\in A$, we add two 3-vertex paths to $G$, whose central
  vertices are adjacent to $x_{4i}$ and $x_{4i+1}$, (the union of the
  two 3-vertex paths is called \emph{a sculpture on $P_A$ at the
    $i$-th location}) and
  \item if $i\in B$, we add two 3-vertex paths to $G$, whose central
  vertices are adjacent to $y_{4i}$ and $y_{4i+1}$ (the union of the
  two 3-vertex paths is called \emph{a sculpture on $P_B$ at the
    $i$-th location}).
\end{itemize}
This is depicted in Figure~\ref{fig:test} (bottom), with $\ell=2$, $A=\{1\}$ and
$B=\{2\}$, with the sculpture on $P_A$ represented in red and the
sculpture on $P_B$ represented in blue.

\begin{observation}\label{obs:deco}
If $G$ is a unit-disk graph embedded in the plane, with an $r$-corridor
$C$ which is at Euclidean distance at least 3 from the vertices at
distance at least 3 from $C$ in $G$, and two disjoint sets
$A,B\subseteq\{1,\ldots,\ell\}$ with $r\ge 4\ell+8$, then the graph obtained from $G$ by
decorating the walls of $C$ with $A$ and $B$ is a unit-disk graph.
\end{observation}

\begin{proof}
This follows from the definition of a corridor: the purpose of the vertices $z_1$ and
$z_{r+1}$ is to make sure that the line segments $[x_1,x_{r+1}]$ and
$[y_1,y_{r+1}]$ are at Euclidean distance at least $\sqrt{3}$ from
each other, which allows one of the two 3-vertex paths to be placed
inside the corridor while the other is placed outside, for each
location $i$, as illustrated in Figure~\ref{fig:test} (bottom).
\end{proof}

It remains to prove that when $A$ and $B$ are not disjoint, if the corridor is
sufficiently tight (and the walls sufficiently regular, as a
consequence), then the graph obtained by decorating the walls with $A$
and $B$ is \emph{not} a unit-disk graph anymore. For this we will
need to assume that $\ell$ is sufficiently small compared to the size
of the corridor (but
still linear in this size).

\begin{lemma}\label{lem:deco}
Let $G'$ be the unit-disk graph with $O(k^{1+\varepsilon})$ vertices from
Observation~\ref{corridorsubst2}. Then there
is $\ell=\Omega(k/\log k)$ such that for any sets $A,B\subseteq\{1,\ldots,\ell\}$, the graph obtained from $G'$ by decorating the walls of each of the $\Theta \left (\tfrac{k}{\log k} \right )$-corridors of $G'$ with $A$ and $B$ is
a unit-disk graph if and only if $A$ and $B$ are disjoint.
\end{lemma}

\begin{proof}
  By Observation~\ref{obs:deco}, we only need to prove the only if direction.
  Fix any unit-disk embedding of the graph $G''$ obtained by
  decorating the walls  of each of the $\log k$ many $r$-corridors of $G'$
  with $A$ and $B$, where $r=\Theta \left (\tfrac{k}{\log k}\right )$.
By Observation~\ref{corridorsubst2} applied to the restriction of the
unit-disk embedding to $G'$, some $r$-corridor $C_i$  is
$O(k^{-1}\log k)$-tight, and in particular the two walls of $C_i$ are
$O(k^{-1}\log k)$-tight and thus $O(1)$-regular. Let
$P_A=x_1,\ldots,x_{r+1}$ and $P_B=y_1,\ldots,y_{r+1}$ be the walls of $C_i$. By Lemma~\ref{lem:regpath} there exists a constant $\alpha>0$ such that $P_A$
and $P_B$ are $\left (1/10,\alpha \tfrac{k}{\log k} \right )$-regular, so the first
$\Theta(k/\log k)$ points of each of $P_A$ and $P_B$ are at Euclidean
distance at most $1/10$
from their ideal points on the two line segments connecting the endpoints
of $P_A$ and the endpoints of $P_B$. As these two line segments lie at
distance at most 2 apart, this does not leave enough space to place a
sculpture on $P_A$ and a sculpture on $P_B$ at the same location $i$,
as
two vertices from these sculptures would be at distance at most 1 from
each other
(while these vertices are non adjacent in the graph).
\end{proof}

The graph $G'$ from Lemma~\ref{lem:deco} is illustrated in Figure~\ref{fig:resume}.

\subsection{Disjointness-expressivity of unit-disk graphs}

We are now ready to prove the main result of this section.

\begin{theorem}\label{thm:udg}
  For any $\delta>0$, the class of unit-disk graphs is
  $(O(\log n),1-\delta)$-disjointness-expressing.
\end{theorem}
\begin{proof}
Let $N$ be a natural integer and $A,B\subseteq\{1, \ldots, N\}$. 
Let $k$ be such that $N=\Theta \left
  (\tfrac{k}{\log k} \right )$ and such that  we can apply Lemma~\ref{lem:deco} with $\ell=N$, which gives us a graph $G'_{A,B}$ where each of the corridors is decorated with $A$ and $B$. 
  Observe that $G'_{A,B}$ contains $\log k$ corridors, each
of length $\Theta\left (\tfrac{k}{\log k} \right)$. Before decoration, the graph has  $O(k^{1+\varepsilon})$ vertices so, even after decorating the walls
of the corridors, the resulting graph $G'_{A,B}$ still contains
$O\left (k^{1+\varepsilon} \right)$ vertices. It follows that the graphs have at most $O((N\log
N)^{1+\varepsilon})=O(N^{1+\delta})$ vertices for any
$\delta>\varepsilon$.

For a corridor $C$ in $G'_{A,B}$, let $P_A$ be the wall of $C$
decorated with $A$, and let $P_A^-$ be $P_A$ minus its two endpoints
(i.e., if $P_A$ consists of the path $x_1,\ldots,x_{r+1}$ plus 
decorations, then $P_A^-=P_A-\{x_1,x_{r+1}\}$). We say that $P_A^-$ is
a \emph{reduced decorated wall}, with endpoints $x_2$ and $x_r$.
We define $L(A)$ to be the subgraph of  $G'_{A,B}$ induced by the
union of the reduced decorated walls $P_A^-$ in each of the corridors
of $G'_{A,B}$. We then define the set of special vertices $S$ as the
union of all endpoints of the reduced decorated walls $P_A^-$. The
graph $R(B)$ is then defined as the subgraph of $G'_{A,B}$ induced by
$(V(G'_{A,B})\setminus L(A))\cup S$.

By construction,
$G'_{A,B}$ is the graph $g(L(A), R(B))$ obtained from $L(A)$ and $R(B)$
by gluing them along $S$.
As there are $\log k=O(\log N)$ corridors and the closed neighborhood
of each vertex of $S$ contains 4 vertices, the size
of $N[S]$ is  $O(\log N)$. Moreover, the subgraph induced by the
closed neighborhood of each vertex of $S$ is independent of $A$ and
$B$, and as the vertices of $S$ are pairwise at distance at least 4
apart, the subgraph induced by the closed neighborhood of $S$ is the
disjoint union of all subgraphs induced by $N[s]$, for $s\in S$, which is independent of $A$
and $B$.
The fact that $G'_{A,B}$ is a unit-disk graph if and only if $A$ and $B$ are disjoint is a direct consequence of Lemma~\ref{lem:deco}.
\end{proof}

Using Theorem~\ref{thm:dex}, together with Corollary~\ref{cor:nlogn}, we immediately deduce the following.

\begin{theorem}
The local complexity of the class of unit-disk
graphs is $O(n\log n)$ and $\Omega(n^{1-\delta})$ for any $\delta>0$.
\end{theorem}

\section{Open problems}\label{sec:ccl}
 In this paper we have obtained a number of optimal (or close to
 optimal) results on the local complexity of geometric graph
 classes. Our proofs are based on a new notion of rigidity. It is
 natural to ask which other graph classes enjoy similar
 properties. A natural candidate is the class of segment graphs
 (intersection graphs of line segments in the plane), which have
 several properties in common with unit-disk graphs, in particular the
 recognition problems for these classes are $\exists
 \mathbb{R}$-complete (i.e., complete for the existential theory of
 the reals) and the minimum bit size for representing an embedding of
 some of these
 graphs in the plane is at least exponential in $n$. We believe
 that the local complexity of segment graphs (and that of the more
 general class of string graphs) is at least polynomial in $n$. More generally, is it true that all classes of graphs for which the
 recognition problem is hard for the
 existential theory of the reals have polynomial local complexity?

 It might also be interesting to investigate the smaller class of 
 \emph{circle graphs} (intersection graphs of chords of a circle).
 The authors of~\cite{JMR23} proved that the closely related class of
 permutation graphs has logarithmic local complexity. It is quite
 possible that the same holds for circle graphs. See \cite{JMR22} for
 results on interactive proof labeling schemes for this class and
 related classes.

 \medskip

 We proved that 1-planar graphs have local complexity $\Theta(n)$.
What can we say about the local complexity of other graph classes
defined with constrained on their drawings in the plane? For instance
is it true that for every $k\ge 2$, the local complexity of the class of
graphs with queue number at most $k$ is polynomial? What about graphs
with stack number at most $k$?

\medskip

We have given the first example of  non-trivial hereditary (and even
monotone) classes of local complexity $\Omega(n)$. Can this be
improved? Are there hereditary (or even monotone) classes of local
complexity $\Omega(n^{c})$ for $c>1$?

\medskip

As a final remark, it was suggested by one of the reviewers that the lower bound techniques used in the
paper can be used to prove lower bounds for the recognition problem in
the studied classes in different models of
computation, namely in the \textsf{CONGEST} model of distributed
computing where messages exchanged by nodes have logarithmic size, and
in the streaming model.

\bigskip

 \paragraph{\bf Acknowledgement} The authors would like to thank the
 reviewers of an earlier version of the paper for their corrections
 and helpful suggestions.

\bibliographystyle{plain}
\bibliography{local.bib}

\appendix
\section{Proofs from Section~\ref{sec:stri}}
\label{sec:prstri}

\lemcurv*

\begin{proof}
Consider a $(k,\delta)$-almost-equilateral embedding of a stripe
$S_\ell$ with  $k \gg \delta$. Up to exchanging the roles of $u_i$ and
$v_i$, the Euclidean distance between $u_0$
and $u_\ell$ is minimized when the distances $d_2(u_i,u_{i+1})$ ($0\le i
\le \ell-1$) are minimized and the angles $\angle u_{i-1}u_{i}u_{i+1}$  ($1\le i \le \ell-1$) are maximized. Since all the constraints on these distances and angles are
local, in an embedding minimizing $d_2(u_0,v_\ell)$ all these lengths
are equal, and all these angles are equal. It follows that the Menger
curvature of all consecutive triples is the same (and only depends on
$k$ and $\delta$, since the local constraints only depend on $k$ and
$\delta$).
\end{proof}

\lemcurvv*

\begin{proof}
  Consider a $(k,\delta)$-almost-equilateral embedding of $S_\ell$ where
$d_2(u_0,v_\ell)\le m+1$.
The second part of the statement simply follows from the fact that in
any $(k,\delta)$-almost-equilateral embedding of $S_\ell$, the
diameter of the corresponding point set is at most $(k+\delta)\ell$,
so we only need to prove the first part of the statement. By Lemma~\ref{lem:curv}, for each $0\le i<j\le \ell$, the Euclidean distance
between $u_i$ and $u_j$ is at least the distance between the endpoints
of some
circular arc of radius $\rho$ that subtends an angle $\tfrac{j-i}\ell
\cdot \pi$, that is $d_2(u_i,u_j)\ge 2 \rho \sin(\tfrac{j-i}\ell
\cdot \tfrac\pi{2})$. It follows that for any $0\le i \le \ell$,  $d_2(u_0,u_i)\ge 2 \rho
\sin(\alpha)$ and $d_2(u_i,u_\ell)\ge 2 \rho
\cos(\alpha)$, where 
$\alpha:=\tfrac{i\pi}{2\ell}$. On the other hand, note that
$d_2(u_0,v_\ell)\le m+1\le 2\rho+O(k)$, where the second inequality
follows from the definition of $\ell$. By Apollonius' theorem,
\begin{eqnarray*}
  d_2(u_i,c)^2 & = &
                   \tfrac12(d_2(u_0,u_i)^2+d_2(u_i,v_\ell)^2)-d_2(u_0,c)^2\\
  & \ge & 2 \rho^2 \sin^2 \alpha+2 \rho^2 \cos^2 \alpha-2k\rho \cos
          \alpha +2k^2-(\rho^2+O(\rho k)+O(k^2))\\
  & \ge & \rho^2-2k\rho -O(k^2)\\
          & \ge & \rho^2(1-O(k/\rho+k^2/\rho^2))
\end{eqnarray*}

\begin{figure}[htb]
\centering
\includegraphics[scale=0.7]{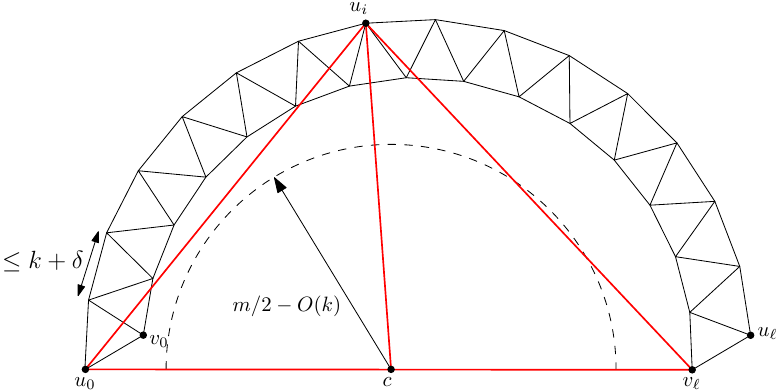}
\caption{Illustration of the proof of Lemma \ref{lem:curv2}.}
\label{fig:stripe3}
\end{figure}
It follows that \[d_2(u_i,c)\ge \rho \cdot (1-O(k/\rho+k^2/\rho^2))\ge
  \rho -O(k+k^2/\rho)\ge \rho-O(k),\]
where the final inequality follows from the fact that $k^2/\rho=O(k^2/m)=O(\delta)=o(k)$.
A
similar computation applies to
the points $v_i$, which concludes the proof.
\end{proof}

\section{Proofs from Section~\ref{sec:qrg}}
\label{sec:prqrg}

\lemtk*

\begin{proof}
Consider a unit-disk graph embedding of $T_n$, and let $P$ be the
$(8n+16)$-gon corresponding to the vertices of $C$. We denote by $P_1$
and $P_2$ the two polygonal subchains of $P$ with endpoints $c$ and
$c'$.  Let $R_1$ be the
region bounded by $P_1$ and the line-segment $[c,c']$, and let
$R_2$ be the
region bounded by $P_2$ and the line-segment $[c,c']$. It might be
the case that one of the two regions contains the other, but in any case
the region $R$ bounded by $P$ is contained in the union of $R_1$ and
$R_2$.

Note that $T_n$ has an independent set $S$ of size at least
$(2n+2)^2$ (the set of red vertices, or the set of blue vertices in
Figure~\ref{fig:tn}). Discard the vertices of $S$ that are at
Euclidean distance at most
$\tfrac12$ from $P$ or from the line-segment $[c,c']$ (there are at
most $O(n)$ such vertices), and let $S'$ be the resulting independent
set (of size $(2n+2)^2-O(n)$). The disks of radius $\tfrac12$ centered
in the points corresponding to
$S'$ are pairwise disjoint, and all contained in $R$, and thus
contained in $R_1$ or $R_2$. It follows that one of $R_1$ or $R_2$
(say $R_1$), contains at least $\tfrac12 (2n+2)^2-O(n)=2(n+1)^2-O(n)$ disjoint disks of radius
$\tfrac12$, and thus has area $A(R_1)\ge
\tfrac{\pi}{4}(2(n+1)^2-O(n))=\tfrac{\pi}{2}(n+1)^2-O(n)$. Let $d:=d_2(c,c')$. Note that $R_1$ and
$R_2$ have perimeter at most $4n+8+d$, and thus by the isoperimetric
inequality in the plane, $A(R_1)\le \tfrac1{4\pi}(4n+8+d)^2$. It follows
that $(4n+d+8)^2\ge 2\pi^2 (n+1)^2-O(n)$, and thus $d\ge (\pi\sqrt{2}-4) n-O(\sqrt{n})=\Omega(n)$, as desired.
\end{proof}

\section{Proofs from Section~\ref{sec:tab}}
\label{sec:prtab}
  
\lemmidtri*

\begin{proof}
  Let $y$ be the length of the altitude of $abc$ passing through
  $c$. Note that $y$ is maximized when $d_2(a,c)=d_2(b,c)$, and thus by Pythagoras' theorem, $y\le\sqrt{\delta x/2}+O(x^{-1/2})$. Note that subject to
  the conditions above, $\min\{d_2(a,c)- d_2(a,c'), d_2(b,c)- d_2(b,c')\}$ is
  maximized when $d_2(c,c')= \rho$.

  \begin{figure}[htb]
\centering
\includegraphics[scale=1]{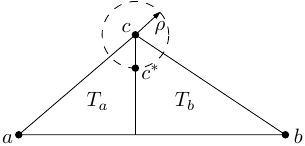}
\caption{Illustration of the proof of Lemma \ref{lem:midtri}.}
\label{fig:triangle2}
\end{figure}

  Let $c^*$ be the point of the
altitude  of $abc$ passing through
$c$ such that $d_2(c,c^*)=\rho$. See Figure \ref{fig:triangle2} for an
illustration. This altitude divides $abc$ into two
triangles, say $T_a$ containing $a$ and $T_b$ containing $b$. If
$c'\in T_a$ then $d_2(b,c')\ge d_2(b,c^*)$, and if $c'\in T_b$ then
$d_2(a,c')\ge d_2(a,c^*)$. Thus it suffices to prove that $d_2(a,c)-
d_2(a,c^*)\le \rho \sqrt{2\delta/x}+O(1/x)$ and $d_2(b,c)-
d_2(b,c^*) \le \rho \sqrt{2\delta/x}+O(1/x)$. The two cases being symmetric, we only
prove the former. Let $z=d_2(a,c)$. By Pythagoras' theorem,
\[d_2(a,c^*)^2=z^2-2\rho y+\rho^2= z^2\left(1-\frac{2\rho y}{z^2}+\frac{\rho^2}{z^2}\right).
\]
As $z=\Omega(x)$, it follows that 
\[
  d_2(a,c^*)=z\sqrt{1-\frac{2\rho y}{z^2}+\frac{\rho^2}{z^2}} = z
  \left(1-\frac{\rho y}{z^2}+O(x^{-2})\right)=z- \rho\sqrt{2\delta/x}+O(x^{-1}).\]
So, $d_2(a,c)-
d_2(a,c^*) \le \rho \sqrt{2\delta/x}+O(1/x) =O(\sqrt{\delta/x})$, as desired.
\end{proof}

\cormidtri*
\begin{proof}
Let $x:=d_2(a,b)$, and let $c\in \mathbb{R}^2$ be such that 
$d_2(a,c)=|P_{\mathbf{x}0}|$ and $d_2(b,c)=|P_{\mathbf{x}1}|$. This
defines (at most) 2 points 
which are symmetric with respect to the line-segment $[a,b]$, and we
define $c$ as the one which is closest from $z_\mathbf{x}$. See Figure \ref{fig:triangle3} for an
illustration. The path $P_\mathbf{x}$ is $\delta$-tight so $d_2(a,c)+d_2(c,b)= x+\delta$.
As $c'$ is a midpoint of
$P_\mathbf{x}$, $|d_2(a,c)-d_2(b,c)|\le 1$. Hence, in order to apply Lemma~\ref{lem:midtri}, which directly gives us the desired result, 
we only need to prove that $\rho:=d_2(c,c')=O(1)$.

 \begin{figure}[htb]
\centering
\includegraphics[scale=1]{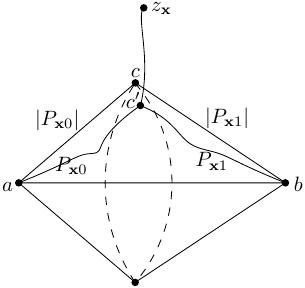}
\caption{Illustration of the proof of Corollary \ref{cor:midtri}.}
\label{fig:triangle3}
\end{figure}

Let $m_i^*$ denote the minimum distance between $z$ and $c'$ in a
unit-disk embedding of $G_i$ and observe that we have $m_i^* \leq d_2(z,c)
\leq d_2(z,c')$. By definition, the path in $G_{i+1}$ linking $z$
and $c'$ has length $\lceil m_i^* \rceil$ so as $G_{i+1}$ is a
unit-disk graph, $d_2(z,c') \leq \lceil m_i^* \rceil$.
Thus $d_2(z,c) \leq d_2(z, c') \leq \lceil d_2(z, c) \rceil$.

Let $R_a$ be the disk of radius $d_2(a,c)$ centered in $a$ and let
$R_b$ be the disk of radius $d_2(b,c)$ centered in $b$.
Let $R_z$ be the disk of radius $\lceil d_2(z,c)\rceil
\le d_2(z,c)+1$ centered in $z$. Note that since $G_0$ is a unit-disk graph, $c'$ lies in
$R_a\cap R_b\cap R_z$. Recall that $z$ lies at distance $\Omega(k)=\Omega(x)$ from $c$
and at distance at most $\tfrac12$ from the perpendicular
bisector of the line-segment $[a,b]$. It then follows from Pythagoras'
theorem that $\rho=d_2(c,c')\le 1+O(x^{-1})=O(1)$, as
desired.
\end{proof}

\lemfixpath*

\begin{proof}
For the base case of the induction, we start by recalling that  in any
unit-disk embedding of $G_0$, $P$ is 1-tight. Assume now that $i\ge 1$ and the
result holds in $G_{i-1}$. Fix a unit-disk embedding of $G_i$ and
consider the restriction of this unit-disk embedding to
$G_{i-1}$. By the induction hypothesis there is $\mathbf{x}\in
\{0,1\}^{i-1}$ such that $|P_\mathbf{x}|\le d_2(a,b)+\delta$ with
$\delta\le \alpha\cdot 2^{i-1}\cdot k^{2^{-i+1}-1}$, for some
constant $\alpha>0$, where $a$ and $b$ denote the endpoints of $P_\mathbf{x}$. Let $c'$ be the midpoint of $P_\mathbf{x}$ used
to split it into $P_{\mathbf{x}0}$ and $P_{\mathbf{x}1}$  in $G_i$. As $d_2(a,b)\ge k2^{-i}$, we have
\[\sqrt{\delta/d_2(a,b)}\le\frac{\alpha^{1/2}\cdot 2^{i/2}\cdot
    k^{2^{-i}-1/2}}{k^{1/2}2^{-i/2}}\le  \alpha^{1/2}\cdot 2^{i}\cdot
    k^{2^{-i}-1},
  \]
  By Corollary~\ref{cor:midtri}, we have that $|P_{\mathbf{x}0}|\le d_2(a,c')+O(\sqrt{\delta/d_2(a,b)})$, or  $|P_{\mathbf{x}1}|\le d_2(b,c')+O(\sqrt{\delta/d_2(a,b)})$.
  By taking $\alpha$ sufficiently large
  compared to the implicit
constant in the $O(\sqrt{\delta/d_2(a,b)})$ term from
Corollary~\ref{cor:midtri}, we obtain that
$|P_{\mathbf{x}0}|\le d_2(a,c')+\alpha\cdot 2^{i}\cdot k^{2^{-i}-1}$, or
$|P_{\mathbf{x}1}|\le d_2(b,c')+\alpha\cdot 2^{i}\cdot k^{2^{-i}-1}$ as desired.
\end{proof}

\lemregpath*

\begin{proof}
 Note that
  $\ell-\delta \le d_2(v_1,v_{\ell+1}) \le \ell$. We place $\ell+1$ evenly
spaced points $u_1,\ldots,u_{\ell+1}$ on the line-segment $[v_1,v_{\ell+1}]$ with
$u_1=v_1$, $u_{\ell+1}=v_{\ell+1}$. For 
$2\le i \le \ell+1$, the distance $d_2(u_i,v_i)$ is maximized when $i=
\left \lceil \frac{\ell+1}{2} \right \rceil$, and is at most the distance between the
middle of the line-segment $[v_1,v_{\ell+1}]$ and a point $c$ at Euclidean
distance $\ell/2$ from $v_1$ and $v_{\ell+1}$. See Figure
\ref{fig:ellipse} for an illustration. By Pythagoras' theorem, this distance
is at most
$\gamma:=\sqrt{\ell^2/4-(\ell-\delta)^2/4}=\sqrt{\ell\delta/2}+O(\ell^{-1/2})$,
as desired.

\begin{figure}[htb]
\centering
\includegraphics[scale=1.2]{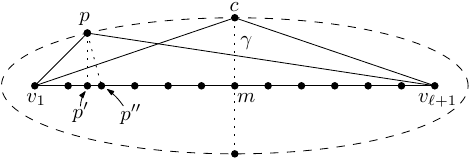}
\caption{Illustration of the proof of Lemma \ref{lem:regpath}.}
\label{fig:ellipse}
\end{figure}

For the second part of the lemma, we observe that each vertex $v_i$
satisfies $d_2(v_i,v_1)+d_2(v_i,v_{\ell+1})\le \ell$, so each such
point is contained in the region bounded by the ellipse with foci
$v_1$ and $v_{\ell+1}$, width $\ell$, and height at most
$2\gamma$ (by the preceding paragraph). Let $m$ denote the midpoint of
the line segment $[v_1,v_{\ell+1}]$. Let $p$ be a point on the
ellipse, and let $p'$ be the projection of $p$ on the line
$(v_1,v_{\ell+1})$. By the standard description of an ellipse in Cartesian
coordinates, if $x:=d_2(m,p')$ then
$y:=d_2(p,p')=\tfrac{2\gamma}{\ell}\sqrt{\ell^2/4-x^2}$. If we
take $(1-\alpha)\ell/2\le x \le \ell/2$ for some $\alpha>0$, then
$y\le \gamma\sqrt{1-(1-\alpha)^2}=\lambda' \gamma$ for
$\lambda':=\sqrt{2\alpha-\alpha^2}$ independent of $\ell$ and
$\delta$. Let $p''$ be the point of $[v_1,v_{\ell+1}]$ at distance
$\tfrac{\ell-\delta}{\ell}d_2(v_1,p)$ from $v_1$ (so that if $v_i$
coincides with $p$, then $u_i$ coincides with $p''$). Then $d_2(p',p'')\le
d_2(v_1,p)-\tfrac{\ell-\delta}{\ell}d_2(v_1,p)=\tfrac{\delta}\ell
d_2(v_1,p)$. It then follows from Pythagoras' theorem that
$d_2(p,p'')\le \sqrt{\lambda'^2\gamma^2+\tfrac{\delta^2}{\ell^2}
d_2(v_1,p)^2}$. If $d_2(v_1,p)\le \alpha \ell$ then $d_2(p,p'')\le
\sqrt{\lambda'^2\gamma^2+\alpha^2\delta^2}\le
\sqrt{\lambda'^2\gamma^2+\alpha^2}$ since $\delta\le 1$. By
considering $p$ above as the location of $v_i$ ($1\le i \le \alpha
\ell$) that maximizes its distance to $u_i$ (whose location
corresponds to $p''$ above), we conclude that $d_2(v_i,u_i)\le
\sqrt{\lambda'^2\gamma^2+\alpha^2}=
\sqrt{(2\alpha-\alpha^2)\gamma^2+\alpha^2}$, as desired.
\end{proof}

\needspace{8\baselineskip} 

\end{document}